\pdfoutput=1
 \pdfoutput=1
 \pdfoutput=1
 \pdfoutput=1
 \pdfoutput=1
\documentclass[journal,draftcls,onecolumn,12pt,twoside]{IEEEtranTCOM}
\normalsize

\ifCLASSINFOpdf
 \else
\fi

\usepackage[normalem]{ulem}
\usepackage{cite}
\usepackage{graphicx}
\usepackage{hyperref}
\usepackage{cleveref}
\usepackage{amsmath}
\usepackage{amssymb}
\usepackage{bm}
\usepackage{algorithmic}
\usepackage{algorithm}
\usepackage{array}
\usepackage{booktabs}
\usepackage{multirow}
\usepackage{makecell}
\usepackage[justification=centering]{caption}
\newtheorem{thm}{Theorem}

\begin{document}
\title{Achievable Rate of Multi-Antenna WSRNs with EH Constraint in the presence of a Jammer}

\author{Minhan~Tian$^\ast$ ~Guiguo~Feng$^\dagger$ ~Wangmei~Guo$^\ast$ ~Jingliang~Gao$^\ast$ ~Yongkang~Li$^\ast$
\thanks{The work is supported by NSFC with No.61701375.}
\thanks{The author with $^\ast$ is with the State Key Lab of Integrated Services Networks, Xidian University, Xi'an, China (e-mail: mhtian@stu.xidian.edu.cn,wangmeiguo@xidian.edu.cn,jlgao@xidian.edu.cn,liyongkang@stu.xidian.edu.cn)}% <-this % stops a space
\thanks{The author with $^\dagger$ is with the 20th Research Institute of China Electronics Technology Group Corporation, Xi'an, China (e-mail: fengguiguo@163.com)}}% <-this % stops a space}


\markboth{IEEE Transactions on Communications}%
{Submitted paper}



\maketitle

\begin{abstract}
%\boldmath
In this paper, the rate-energy region is studied for the wireless sensor relay network (WSRN) with energy harvesting in the presence of a jammer. In the model, a source communicates to a destination equipped with a single antenna with energy harvesting constraint through a multi-antenna cooperative relay under beamforming. Meanwhile, there is a jammer intended to disturb the communication. The relay works in half-duplex mode and knows all the channel state information (CSI). When beamforming is employed at the relay, the network can be modeled as an equivalent Gaussian arbitrarily varying channel (GAVC). We characterize the achievable rate-energy region. Since the problem is non-convex, we present three methods to transform it into a semi-definite programming problem (SDP), and the closed-form expression for two special boundary points of the rate-energy region is obtained. Finally, the simulations show the rate-energy region and the anti-jamming performance of the proposed scheme.
\end{abstract}


% \fi
\IEEEpeerreviewmaketitle

\section{Introduction}
\IEEEPARstart{T}{he} lifetime of energy-constrained communication networks becomes one of the most challenging issues with the large-scale application of a smart city, especially for massive sensor networks. To extend the network lifetime,
  both aspects of improving network transmission efficiency and harvesting energy might be considered. For the first point, collaborative beamforming (CB) enhances the network utilization by employing the idle node and weighting received signals \cite{t1,t2,b3,c3,c4}. It can effectively extend the signal transmission coverage, improve transmission efficiency \cite{t2}, save network resources, decrease jamming and strengthen security. Additionally, energy harvesting (EH) provides long-term energy support for wireless sensor networks (WSNs)\cite{e4}. Significantly, energy harvesting WSNs (EHWSNs) harvest environmental energy sources such as solar, wind, piezoelectric, and radio frequency signals\cite{d4}. Therefore, it is in line with the concept of environmental protection. For this reason, it is presented as a green communication method for prolonging the lifetime, and is attracting much attention in academia as well as industry \cite{t3,t4,b5,c5,c6}.

Naturally, since the energy harvesting node has higher channel quality requirements than the target node, CB can be applied in EHWSNs. Under this circumstance, more gain will be obtained when combining CB and EH, and some related works are done in \cite{t5,t6,t7}. Jianli Huang \textit{et al.} studied the optimization problem of beamforming under the constraints of relay transmission power and EH \cite{t5}. Xuecai Bao \textit{et al.} proposed a software-defined energy harvesting wireless sensor network (EHWSN) architecture to maximize the signal-to-noise ratio (SNR) for CB with sidelobe constraints by allocating the transmission power of sensor nodes \cite{t6}. However, both of them did not consider the case with jamming source.
In practice, the transmission in WSNs is easily attacked or interfered due to its shared property \cite{t8}. This will cause transmission failures and retransmissions in which the later will degrade network performance and drain the battery of the device rapidly \cite{d8}. Since jamming is an important active attack mode in wireless networks, the anti-jamming problem is one of the aims that should be considered in the designed model.
The various anti-jamming strategies were proposed to eliminate the jamming and improve the legitimate transmission rate for traditional WSNs \cite{t9,t10,c10,c11,c12,c13}, such as spreading spectrum, beamforming and interference alignment (IA). However, there is less work on anti-jamming beamforming scheme for WSNs with EH. Soheil \textit{et al.} \cite{t7} investigated the problem of joint EH time allocation and distributed beamforming in the presence of interference, which is the special case of jamming when the transmission power is higher than the jamming power. Besides, the multiple-antenna relay was missed. Given that large bandwidth is required for spreading spectrum, and that all the channel state information (CSI) should be available for all nodes in IA, we consider beamforming technology at multiple-antenna relay against the jammer for WSN with EH.

In this paper, we consider a dual-hop relay network, which consists of a source node, a multi-antenna cooperative relay, a legitimate receiver, an EH node and a jammer. The relay receives the sum signal from the source and the jammer in one time-slot, and then forwards the received signal to the destination in the
next time-slot with beamforming. With the linear beamforming scheme at the relay, this network can be modeled as an equivalent Gaussian arbitrarily varying channel (GAVC). We mainly study the maximum achievable rate-energy region of EHWSN based on cooperative beamforming at multi-antenna relay under the constraints of EH and sum power. Specifically, we give the formulation of the optimization problem to maximize the transmission rate, which truns out to be non-convex. We propose three methods to transform it into a semi-definite programming based on the aspects of system stability, lower computational complexity, and the combination of both. Then the closed-form expression for two special boundary points of the rate-energy region is obtained. Finally, we present simulation results. Interestingly enough, the simulation results show that the achievable rate is positive even when the power of the jammer is larger than that of the legitimate source by employing the optimal linear beamforming matrix.
The main contributions of this paper are listed as follows.
\begin{itemize}
  \item Anti-jamming ability: Collaborative beamforming is effective to reduce the impact of jamming attacks and makes full use of jamming information to improve the signal to interference plus noise ratio (SINR).
  \item Prolonging the lifetime: The collaborative beamforming and energy harvesting networks reduce the node energy consumption rate from the aspects of increasing the transmission rate and harvesting energy, respectively.
  \item Variety of processing methods: We consider the processing scheme from non-convex problem to convex problem from three aspects, namely stability, complexity and combination of the two.
\end{itemize}

The rest of this paper is organized as follows. We present the system model, the EH constraint and the expression of capacity in Section \ref{sec:model}. Next, the optimization procedure schemes are shown in Section \ref{sec:The Optimization Method} and the performance analysis is given in Section \ref{simulations}. Finally, Section \ref{sec:conclusion} concludes the entire paper.

\emph{\textbf{Notation}}: Scalars are denoted by lower-case letters, e.g., $x$, and bold-face lower-case letters are used for column vectors, e.g., ${\bf{x}}$, and bold-face upper-case letters for matrices, e.g., ${\bf{X}}$.  ${\left(  \cdot  \right)^ * }$, ${\left(  \cdot  \right)^T}$, ${\left(  \cdot  \right)^\dag }$ and $tr(\cdot)$ denote the conjugate, transpose, Hermitian transpose and trace, respectively. ${\rm{diag}}\left( {{x_1},{x_2},...,{x_n}} \right)$ denotes the diagonal square matrix with ${x_1},{x_2},...,{x_n}$ as the diagonal elements and $\left\|  \cdot  \right\|$ denotes the Euclidean norm.  ${{\bf{I}}_n}$ is the $n$-dimensional identity matrix and $\mathbb{E}\left(  \cdot  \right)$ is the expectation operation. The $\otimes$ denotes the Kronecker product.

\section{System Model}\label{sec:model}
\begin{figure}[H]
   \label{figure 1}
	\centering
	\includegraphics[width=3in]{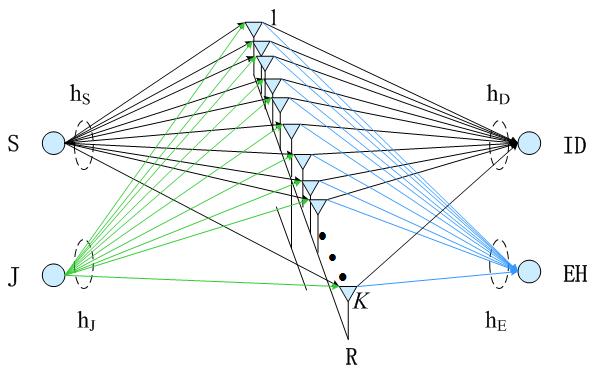}
	\caption{Multiple-antenna Relay Model}
    \label{figure}
\end{figure}
As illustrated in Fig.\ref{figure}, the system consists of an EHWSN with a source node $S$, a jammer $J$, an energy harvesting node EH, an information receiver ID, and a multi-antenna relay node $R$ which is equipped with $K$ antennas. This network model is developed with the following assumptions. Firstly, the transferring scheme is assumed to be simultaneous wireless information and power transfer (SWIPT). Secondly, the direct links between the source (or jammer) and the ID receiver (or EH node) are sufficient to be ignored in the case that they are far away. Thirdly, the flat fading channel is employed in this scene and the perfect synchronization is achieved before the transmission. Finally, the accurate CSI of the network is available at the relay working in half-duplex mode. The specific information transmission process is as follows.

In the first time-slot, the relay $R$ receives the transmitted signals $x_S$, $x_J$ from the source node $S$ with power $P_S$ and the jammer $J$ with power $P_J$, respectively. The signal received at the relay node $R$ can be formed as a $K$-dimensional column vector ${{\bf{y}}_R}$, given by
\begin{align}
\label{case2:1}
{{\bf{y}}_R} = \sqrt{P_S}{{\bf{h}}_S}{x_S} + \sqrt{P_J}{{\bf{h}}_J}{x_J} + {{\bf{z}}_R}.
\end{align}
where, ${{\bf{z}}_R}\sim\mathcal{CN}\left( {0,{\sigma_R}^2{{\bf{I}}_K}} \right)\in\mathbb{C}^{K\times1}$ is the additive Gaussian noise at the relay $R$. Besides, $\mathbf{h}_S=[h_{S, 1}, h_{S, 2}, ..., h_{S, K}]^T\in\mathbb{C}^{K\times1}$ is the channel fading vector from the source node to the relay $R$, in which $h_{S,k}$, $k\in{[1,K]}$ stands for the channel fading coefficient from the source node to the $k$-th antenna at the relay. Moreover, $\mathbf{h}_J=[h_{J, 1}, h_{J, 2}, ..., h_{J, K}]^T\in\mathbb{C}^{K\times1}$ denotes the channel fading vector from the jammer to the relay, where $h_{J,k}$, $k\in{[1,K]}$ indicates the channel fading coefficient from the jammer to the $k$-th antenna at the relay.

In the second time-slot, the signal and wireless energy addressed via beamforming in the relay $R$ are transferred to the ID receiver and the EH receiver, respectively. Assuming that the beamforming matrix is ${\bf{A}}$, the retransmitted signal at the relay is given by
\begin{align}
\label{case2:3}
{{\bf{x}}_R} = \sqrt{P_S}{\bf{A}}{{\bf{h}}_S}{x_S} +  \sqrt{P_J}{\bf{A}}{{\bf{h}}_J}{x_J} + {\bf{A}}{{\bf{z}}_R}.
\end{align}

According to \eqref{case2:3}, the power of the signal transmitted by the relay $R$ can be written as
\begin{align}
\label{case2:4}
{\mathbb{E}}[|{{\bf{x}}_R}{|^2}]=\lVert {\bf{A}}{{\bf{h}}_S}{\rVert^2}{P_S}+\lVert{\bf{A}}{{\bf{h}}_J}{\rVert^2}{P_J}+{\sigma_R}^2tr({\bf{A}}{\bf{A}}^\dag).
\end{align}

Given that the power of the relay $R$ is upper bounded by $P_{R,\max}$, the following condition should be satisfied by ${{\bf{x}}_R}$.
\begin{align}
\label{case2:5}
\mathbb{E}[|{{\bf{x}}_R}{|^2}]\le P_{R,\max}.
\end{align}

The signal received by the receiver node ID is
\begin{align}
\label{case2:6}
{{\bf{y}}_D}&={{\bf{h}}_D}^T {{\bf{x}}_R} + z_D\nonumber\\
&=\sqrt{P_S}{{\bf{h}}_D}^T{\bf{A}}{{\bf{h}}_S}{x_S} + \sqrt{P_J}{{\bf{h}}_D}^T{\bf{A}}{{\bf{h}}_J}{x_J}+{{\bf{h}}_D}^T{\bf{A}}{{\bf{z}}_R}+ z_D.
\end{align}
where, $z_D\sim \mathcal{CN} \left( {0,\sigma _D^2} \right)$ denotes the complex Gaussian noise received at the information receiver ID, which is independent of ${{\bf{z}}_R}$.

Let $x_{S,eq}=\sqrt{P_S}{{\bf{h}}_D}^T{\bf{A}}{{\bf{h}}_S}{x_S}$, $x_{J,eq}=\sqrt{P_J}{{\bf{h}}_D}^T{\bf{A}}{{\bf{h}}_J}{x_J}$, and $z_{eq}={{\bf{h}}_D}^T{\bf{A}}{{\bf{z}}_R}+ z_D$, then $y_D$ is represented as
\begin{align}
\label{case2:7}
{y_D}=x_{S,eq}+x_{J,eq}+{z_{eq}}.
\end{align}

Based on the symmetric condition of the GAVC \cite{t11}, the capacity is
\begin{align}
\label{capacity}
C(\mathbf{A})=\left\{\begin{aligned}\frac{1}{2}\log(1+\text{SINR}),\ \  &\mathbb{E}|x_{S,eq}|^2>\mathbb{E}|x_{J,eq}|^2\\
0,\ \ \ \ \ \ \ \ \ \ &\text{others.}
\end{aligned}\right.
\end{align}
where, $\mathbb{E}|x_{S,eq}|^2>\mathbb{E}|x_{J,eq}|^2$ is the necessary condition to ensure that the deterministic coding capacity of the Gaussian arbitrarily varying channel (GAVC) is nonzero.

The more specific expression of $\text{SINR}$ in \eqref{capacity} is
\begin{align}
\label{case2:8}
\text{SINR}=\frac{|{{\bf{h}}_D}^T{\bf{A}}{{\bf{h}}_S}{|^2}{P_S}}{|{{\bf{h}}_D}^T{\bf{A}}{{\bf{h}}_J}{|^2}{P_J}+\lVert {{\bf{h}}_D}^T{\bf{A}}{\rVert^2}{{\sigma}_R}^2+{{\sigma}_D}^2}.
\end{align}

The harvested power at EH receiver in the second time-slot should meet the EH constraint as follows
\begin{align}
\label{case2:9}
Q_E&=\mathbb{E}[|{{\bf{h}}_E}^T
{{\bf{x}}_R}{|^2}]\nonumber\\
&=|{{\bf{h}}_E}^T{\bf{A}}{{\bf{h}}_S}{|^2}{P_S} + |{{\bf{h}}_E}^T{\bf{A}}{{\bf{h}}_J}{|^2}{P_J} + \lVert{{\bf{h}}_E}^T{\bf{A}}{\rVert^2}{{\sigma}_R}^2\nonumber\\
&\ge Q.
\end{align}
where, $Q$ is the preset threshold of harvesting energy and the energy fading vector from the relay to the energy harvesting node is presented as $\mathbf{h}_E=[h_{E, 1}, h_{E, 2}, ..., h_{E, K}]^T\in\mathbb{C}^{K\times1}$, where $h_{E,k}$, $k\in{[1,K]}$ indicates the energy fading coefficient from the $k$-th antenna at the relay to the energy harvesting node.

A main focus of our research is to obtain the maximum value of $C(\mathbf{A})$ with EH and power constraints. Thus, considering the formulas \eqref{case2:5}, \eqref{case2:8}, \eqref{case2:9} and the monotonic increase of $\log x$, the optimization problem can be performed as
\begin{align}
\label{case2:11}
\underset{\bf{A}}{\max}\ \
&\frac{|{{\bf{h}}_D}^T{\bf{A}}{{\bf{h}}_S}{|^2}{P_S}}{|{{\bf{h}}_D}^T{\bf{A}}{{\bf{h}}_J}{|^2}{P_J}+\lVert{{\bf{h}}_D}^T{\bf{A}}{\rVert^2}{{{\sigma}_R}^2}+{{\sigma}_D}^2}\nonumber \\
\text{s.t.} \ \ &\lVert {\bf{A}}{{\bf{h}}_S}{\rVert^2}{P_S}+\lVert{\bf{A}}{{\bf{h}}_J}{\rVert^2}{P_J}+{\sigma_R}^2tr({\bf{A}}{\bf{A}}^\dag)\le P_{R,\max}\nonumber\\
&|{{\bf{h}}_E}^T{\bf{A}}{{\bf{h}}_S}{|^2}{P_S} + |{{\bf{h}}_E}^T{\bf{A}}{{\bf{h}}_J}{|^2}{P_J} + \lVert{{\bf{h}}_E}^T{\bf{A}}{\rVert^2}{{\sigma}_R}^2\ge Q \nonumber \\
&\frac{{|{{\bf{h}}_D}^T{\bf{A}}{{\bf{h}}_J}{|^2}{P_J}}}{|{{\bf{h}}_D}^T{\bf{A}}{{\bf{h}}_S}{|^2}{P_S}}<1.
\end{align}

\section{The Optimization Method}\label{sec:The Optimization Method}
In this section, several efficient strategies are proposed to simplify \eqref{case2:11} by considering three aspects, system stability, lower computational complexity and the combination. Meanwhile, we discuss how to design an optimal beamforming matrix $\mathbf{A}$ to maximize the capacity $C(\mathbf{A})$ from the source to the ID receiver under the constraints of EH and sum power. Additionally, we study the rate-energy region with a jammer in EHWSN.

\subsubsection{System Stability}
To make the calculation more convenient, considering the stability of the Cartesian product and the formula $vec({\mathbf{A}_1}{\mathbf{A}_2}{\mathbf{A}_3})=({\mathbf{A}_3}^T\otimes {\mathbf{A}_1})\cdot vec({\mathbf{A}_2})$ \cite{t12}, \eqref{case2:11} can be transformed into the following form
\begin{align}
\label{case2:12}
\underset{\bm{\alpha}}{\max}\ \
&\frac{{\bm{\alpha}}^\dag{{\bf{h}}_1}{{\bf{h}}_1}^\dag{\bm{\alpha}}{P_S}}{{\bm{\alpha}}^\dag{{\bf{h}}_2}{{\bf{h}}_2}^\dag{\bm{\alpha}}{P_J}+{\bm{\alpha}}^\dag{{\bf{H}}_1}{{\bf{H}}_1}^\dag{\bm{\alpha}}{{\sigma}_R}^2+{{{\sigma}_D}^2}}\nonumber \\
\text{s.t.} \ \ &{\bm{\alpha}}^\dag \mathbf{\Sigma}{\bm{\alpha}}\le P_{R,\max}\nonumber\\
&{\bm{\alpha}}^\dag{{\bf{h}}_3}{{\bf{h}}_3}^\dag{\bm{\alpha}}{P_S}+{\bm{\alpha}}^\dag{{\bf{h}}_4}{{\bf{h}}_4}^\dag{\bm{\alpha}}{P_J}+{\bm{\alpha}}^\dag{{\bf{H}}_4}{{\bf{H}}_4}^\dag{\bm{\alpha}}{{\sigma}_R}^2\ge Q \nonumber \\
&\varepsilon{\bm{\alpha}}^\dag{{\bf{h}}_1}{{\bf{h}}_1}^\dag{\bm{\alpha}}{P_S}\ge {\bm{\alpha}}^\dag{{\bf{h}}_2}{{\bf{h}}_2}^\dag{\bm{\alpha}}{P_J}.
\end{align}
where, ${{\bf{h}}_1}={{\bf{h}}_S}^\ast\otimes {{\bf{h}}_D}$, ${{\bf{h}}_2}={{\bf{h}}_J}^\ast\otimes {{\bf{h}}_D}$, ${{\bf{h}}_3}={{\bf{h}}_S}^\ast\otimes {{\bf{h}}_E}$, ${{\bf{h}}_4}={{\bf{h}}_J}^\ast\otimes {{\bf{h}}_E}$, ${{\bf{H}}_1}={{\bf{I}}}\otimes {{\bf{h}}_D}$, ${{\bf{H}}_2}={{\bf{h}}_S}^\ast\otimes {{\bf{I}}}$, ${{\bf{H}}_3}={{\bf{h}}_J}^\ast\otimes {{\bf{I}}}$, ${{\bf{H}}_4}={{\bf{I}}}\otimes {{\bf{h}}_E}$, $vec({{\bf{A}}})={\bm{\alpha}}$, $\mathbf{\Sigma}={{\bf{H}}_2}{{\bf{H}}_2}^\dag{P_S}+{{\bf{H}}_3}{{\bf{H}}_3}^\dag{P_J}+{{\sigma}_R}^2{\mathbf{I}}$. In addition, when $\varepsilon \in [0,1)$ approaches 1, it does not affect global optimality. To avoid ${|{{\bf{h}}_D}^T{\bf{A}}{{\bf{h}}_J}{|^2}{P_J}}={|{{\bf{h}}_D}^T{\bf{A}}{{\bf{h}}_S}{|^2}{P_S}}$, which is contrary to the condition that the deterministic coding capacity of the GAVC is non-zero, we may not set $\varepsilon=1$. It is clear that when $\varepsilon$ approaches 1, the larger feasible domain of \eqref{case2:12} is obtained, and it is  more likely to acquire the global optimal solution of the problem.
\subsubsection{lower computational complexity}
One of the advantages of this method is its lower computational complexity obtained by adopting the optimal beamforming matrix structure. Assuming that the singular value decomposition of matrix $[{{\bf{h}}_S},{{\bf{h}}_J},{{\bf{h}}_D},{{\bf{h}}_E}]$ can be written as $[{{\bf{h}}_S},{{\bf{h}}_J},{{\bf{h}}_D},{{\bf{h}}_E}]={\bf{U}}^\ast{\bf{\Omega}}{\bf{V}}^\dag$, where ${\bf{\Omega}}=diag{(\omega_1,\omega_2,...,\omega_r,0,...,0)}$ and ${\bf{U}}=[{\mathbf{U}_1}, {\mathbf{U}_2}]$. $\omega_i (i=1,2,...,r)$ is the positive singular value of the matrix $[{{\bf{h}}_S},{{\bf{h}}_J},{{\bf{h}}_D},{{\bf{h}}_E}]$, where $r$ is the number of non-zero values of the singular solution. Therefore, ${\bf{A}}={\mathbf{U}_1}^\ast{\bf{B}}{\mathbf{U}_1}^\dag+{\mathbf{U}_1}^\ast{\bf{C}}{\mathbf{U}_2}^\dag$ is established, where ${\mathbf{U}_1}$ and ${\mathbf{U}_2}$ represent the former $r$ and the latter $K-r$ columns of ${\bf{U}}$, respectively. Accordingly, \eqref{case2:11} can be rewritten as
\begin{align}
\label{case2:2}
\underset{\bf{B,C}}{\max}\ \
&\frac{|{\mathbf{g}_3}^T{\bf{B}}{\mathbf{g}_1}{|^2}{P_S}}{|{\mathbf{g}_3}^T{\bf{B}}{\mathbf{g}_2}{|^2}{P_J}+\lVert{\mathbf{g}_3}^T{\bf{B}}{\rVert^2}{{{{\sigma}_R}^2}}+\lVert{\mathbf{g}_3}^T{\bf{C}}{\rVert^2}{{{{\sigma}_R}^2}}+{{{\sigma}_D}^2}}\nonumber \\
\text{s.t.} \ \ &{|{\mathbf{g}_4}^T{\bf{B}}{\mathbf{g}_1}{|^2}{P_S}}+{|{\mathbf{g}_4}^T{\bf{B}}{\mathbf{g}_2}{|^2}{P_J}}+\lVert{\mathbf{g}_4}^T{\bf{B}}{\rVert^2}+\lVert{\mathbf{g}_3}^T{\bf{C}}{\rVert^2}\ge Q \nonumber\\
&\lVert{\bf{B}}{\mathbf{g}_1}{\rVert^2}{P_S}+\lVert{\bf{B}}{\mathbf{g}_2}{\rVert^2}{P_J}+tr({\bf{B}}{\bf{B}}^\dag){{\sigma}_R}^2+tr({\bf{C}}{\bf{C}}^\dag){{\sigma}_R}^2 \le P_{R,\max}\nonumber \\
&{|{\mathbf{g}_3}^T{\bf{B}}{\mathbf{g}_2}{|^2}{P_J}}\le\varepsilon{|{\mathbf{g}_3}^T{\bf{B}}{\mathbf{g}_1}{|^2}{P_S}}.
\end{align}
where, ${\mathbf{g}_1}={\mathbf{U}_1}^\dag{{\bf{h}}_S}$, ${\mathbf{g}_2}={\mathbf{U}_1}^\dag{{\bf{h}}_J}$, ${\mathbf{g}_3}={\mathbf{U}_1}^\dag{{\bf{h}}_D}$, ${\mathbf{g}_4}={\mathbf{U}_1}^\dag{{\bf{h}}_E}$.

Because the calculation is based on vector, technically, it is feasible to transform \eqref{case2:2} to the optimization problem as follows
\begin{align}
\label{case2:28}
\underset{\mathbf{b},\bf{c}}{\max}\ \
&\frac{|{\hat{\bf{h}}_1}^T{\bf{b}}{|^2}{P_S}}{|{\hat{\bf{h}}_2}^T{\bf{b}}{|^2}{P_J}+\lVert{\hat{\bf{H}}_3}{\bf{b}}{\rVert^2}{{\sigma}_R}^2+\lVert{\hat{\bf{H}}_3}{\bf{c}}{\rVert^2}{{\sigma}_R}^2+{{{\sigma}_D}^2}}\nonumber \\
\text{s.t.} \ \ &{|{\hat{\bf{h}}_3}^T{\bf{b}}{|^2}{P_S}}+{|{\hat{\bf{h}}_4}^T{\bf{b}}{|^2}{P_J}}+\lVert{\hat{\bf{H}}_4}{\bf{b}}{\rVert^2}{{\sigma}_R}^2+\lVert{\hat{\bf{H}}_4}{\bf{c}}{\rVert^2}{{\sigma}_R}^2\ge Q \nonumber\\
& {\bf{b}}^\dag \mathbf{\Phi} {\bf{b}}+{{\sigma}_R}^2 {\bf{c}}^\dag{\bf{c}}\le P_{R,\max}\nonumber \\
&{|{\hat{\bf{h}}_2}^T{\bf{b}}{|^2}{P_J}}\le\varepsilon{|{\hat{\bf{h}}_1}^T{\bf{b}}{|^2}{P_S}}.
\end{align}
where, $vec({{\bf{B}}})={\bf b}$, $vec({{\bf{C}}})={\bf c}$, ${\hat{\bf{h}}_1}=vec({\mathbf{g}_1}{\mathbf{g}_3}^T)\in\mathbb{C}^{r^2\times1}$, ${\hat{\bf{h}}_2}=vec({\mathbf{g}_2}{\mathbf{g}_3}^T)\in\mathbb{C}^{r^2\times1}$, ${\hat{\bf{h}}_3}=vec({\mathbf{g}_1}{\mathbf{g}_4}^T)\in\mathbb{C}^{r^2\times1}$, ${\hat{\bf{h}}_4}=vec({\mathbf{g}_2}{\mathbf{g}_4}^T)\in\mathbb{C}^{r^2\times1}$, ${\hat{\bf{H}}_3}={\mathbf{g}_3}^T\otimes{\mathbf{I}_r}\in\mathbb{C}^{r^2\times r}$, ${\hat{\bf{H}}_4}={\mathbf{g}_4}^T\otimes{\mathbf{I}_r}\in\mathbb{C}^{r^2\times r}$,
${\bf{\Theta}}={\mathbf{g}_1}{\mathbf{g}_1}^\dag{P_S}+{\mathbf{g}_2}{\mathbf{g}_2}^\dag{P_J}+{{\sigma}_R}^2{\mathbf{I}_r}\in\mathbb{C}^{r\times r}$. Since $\mathbf{\Phi}$ is semi-definite, $\mathbf{\Phi}={\bf{\Omega}}^\dag {\bf{\Omega}}$ is met, and it has the following expression
\begin{align}
\mathbf{\Phi}&= diag\left\{\begin{matrix} \underbrace{\bf{\Theta},\cdots,\bf{\Theta}}_{{r}}\end{matrix}\right\}\in\mathbb{C}^{r^2\times r^2}.\nonumber
\end{align}

\begin{thm}
The optimal beamforming matrix structure is that ${\bf{A}}={\mathbf{U}_1}^\ast{\bf{B}}{\mathbf{U}_1}^\dag+{\mathbf{U}_1}^\ast{\bf{C}}{\mathbf{U}_2}^\dag$.
\end{thm}
\begin{proof}
The proof of this theorem can be seen in the appendix A.
\end{proof}

\subsubsection{The Combination}
Note that stability and lower complexity are two major factors to estimate the model. Hence, combing all of the previous works should arouse our attention. The relative form over this point can be presented as
\begin{align}
\label{case2:30}
\underset{\tilde{\bf{b}},\tilde{\bf{c}}}{\max}\ \
&\frac{{\tilde{\bf{b}}}^\dag{\tilde{\bf{h}}_1}{\tilde{\bf{h}}_1}^\dag{\tilde{\bf{b}}}{P_S}}{\tilde{{\bf{b}}}^\dag{\tilde{\bf{h}}_2}{\tilde{\bf{h}}_2}^\dag{\tilde{\bf{b}}}{P_J}+{\tilde{\bf{b}}}^\dag{\tilde{\bf{H}}_1}{\tilde{\bf{H}}_1}^\dag{\tilde{\bf{b}}}{{{{\sigma}_R}^2}}+{\tilde{\bf{c}}}^\dag{\tilde{\bf{H}}_1}{\tilde{\bf{H}}_1}^\dag{\tilde{\bf{c}}}{{{{\sigma}_R}^2}}+{{{\sigma_D}}^2}}\nonumber\\
\text{s.t.} \ \ &{\tilde{\bf{b}}}^\dag{\tilde{\bf{h}}_3}{\tilde{\bf{h}}_3}^\dag{\tilde{\bf{b}}}{P_S}+{\tilde{\bf{b}}}^\dag{\tilde{\bf{h}}_4}{\tilde{\bf{h}}_4}^\dag{\tilde{\bf{b}}}{P_J}+{\tilde{\bf{b}}}^\dag{\tilde{\bf{H}}_4}{\tilde{\bf{H}}_4}^\dag{\tilde{\bf{b}}}{{{{\sigma}_R}^2}}+{\tilde{\bf{c}}}^\dag{\tilde{\bf{H}}_4}{\tilde{\bf{H}}_4}^\dag{\tilde{\bf{c}}}{{{{\sigma}_R}^2}}\ge Q \nonumber\\
&{\tilde{\bf{b}}}^\dag{\tilde{\bf{H}}_2}{\tilde{\bf{H}}_2}^\dag{\tilde{\bf{b}}}{P_S}+{\tilde{\bf{b}}}^\dag{\tilde{\bf{H}}_3}{\tilde{\bf{H}}_3}^\dag{\tilde{\bf{b}}}{P_J}+{\tilde{\bf{b}}}^\dag{\tilde{\bf{b}}}{{{{\sigma}_R}^2}}+{\tilde{\bf{c}}}^\dag{\tilde{\bf{c}}}{{{{\sigma}_R}^2}} \le P_{R,\max}\nonumber \\
&{\tilde{\bf{b}}}^\dag{\tilde{\bf{h}}_2}{\tilde{\bf{h}}_2}^\dag{\tilde{\bf{b}}}{P_J} \le \varepsilon{\tilde{\bf{b}}}^\dag{\tilde{\bf{h}}_1}{\tilde{\bf{h}}_1}^\dag{\tilde{\bf{b}}}{P_S}.
\end{align}
where, ${\tilde{\bf{h}}_1}={{\bf{g}}_1}^\ast\otimes {{\bf{g}}_3}$, ${\tilde{\bf{h}}_2}={{\bf{g}}_2}^\ast\otimes {{\bf{g}}_3}$, ${\tilde{\bf{h}}_3}={{\bf{g}}_1}^\ast\otimes {{\bf{g}}_4}$, ${\tilde{\bf{h}}_4}={{\bf{g}}_2}^\ast\otimes {{\bf{g}}_4}$, ${\tilde{\bf{H}}_1}={{\bf{I}}}\otimes {{\bf{g}}_3}$, ${\tilde{\bf{H}}_2}={{\bf{g}}_1}^\ast\otimes {{\bf{I}}}$, ${\tilde{\bf{H}}_3}={{\bf{g}}_2}^\ast\otimes {{\bf{I}}}$, ${\tilde{\bf{H}}_4}={{\bf{I}}}\otimes {{\bf{g}}_4}$, $ vec({\tilde{\bf{B}}})={\tilde{\bf{b}}}$, $vec({\tilde{\bf{C}}})={\tilde{\bf{c}}}$. Besides, when $\varepsilon \in [0,1)$ approaches 1, it does not affect global optimality.

\subsection{ The optimal beamforming matrix design with EH constraint}
In this subsection, we will discuss how to solve the above problems specifically. Obviously, the problem \eqref{case2:12} is a non-convex optimization problem which is difficult to solve in the polynomial time. In the following, we show the transformation methods to convert it into a standard convex SDP problem in detail.

To simplify the third constraint in \eqref{case2:12}, the auxiliary variables $u$ and $v$ are introduced. Given that ${\rho_0} = \frac{{{P_J}}}{{{P_S}}}$, ${\rho_1} = \frac{{\sigma_R^2}}{{{P_S}}}$, ${\rho_2} = \frac{{\sigma _D^2}}{{{P_S}}}$, and ${\sigma_R^2}={\sigma_D^2}=1$, the expression about $u$ and $v$ should be $u^2\ge {\bm{\alpha}}^\dag{{\bf{h}}_2}{{\bf{h}}_2}^\dag{\bm{\alpha}}$ and $v^2\ge{\rho_0}{u^2}+{\rho_1}{\bm{\alpha}}^\dag{{\bf{H}}_1}{{\bf{H}}_1}^\dag{\bm{\alpha}}+{\rho_2}$. It is worth noting that the third condition given in \eqref{case2:12} is addressed as the ratio of the equivalent jammer power to the equivalent source power rather than the difference of them, since it is a more robust formulation regardless of the values of the equivalent powers. Based on the assumptions, \eqref{case2:12} can be expressed as
\begin{align}
\label{case2:14}
\underset{\bm{\alpha},v^2,u^2}{\max}\ \
&\frac{{\bm{\alpha}}^\dag{{\bf{h}}_1}{{\bf{h}}_1}^\dag{\bm{\alpha}}}{v^2} \nonumber \\
\text{s.t.} \ \ &v^2\ge{\rho_0}{\bm{\alpha}}^\dag{{\bf{h}}_2}{{\bf{h}}_2}^\dag{\bm{\alpha}}+{\rho_1}{\bm{\alpha}}^\dag{{\bf{H}}_1}{{\bf{H}}_1}^\dag{\bm{\alpha}}+{\rho_2}\nonumber\\
&{\bm{\alpha}}^\dag{{\bf{h}}_1}{{\bf{h}}_1}^\dag{\bm{\alpha}}\ge \frac{{\rho_0} u^2}{\varepsilon}\nonumber\\
&{\bm{\alpha}}^\dag{{\bf{h}}_2}{{\bf{h}}_2}^\dag{\bm{\alpha}}\le u^2\\
&{\bm{\alpha}}^\dag{{\bf{h}}_3}{{\bf{h}}_3}^\dag{\bm{\alpha}}{P_S}+{\bm{\alpha}}^\dag{{\bf{h}}_4}{{\bf{h}}_4}^\dag{\bm{\alpha}}{P_J}+{\bm{\alpha}}^\dag{{\bf{H}}_4}{{\bf{H}}_4}^\dag{\bm{\alpha}}\ge Q \nonumber \\
&{\bm{\alpha}}^\dag{{\bf{H}}_2}{{\bf{H}}_2}^\dag{\bm{\alpha}}{P_S}+{\bm{\alpha}}^\dag{{\bf{H}}_3}{{\bf{H}}_3}^\dag{\bm{\alpha}}{P_J}+{\bm{\alpha}}^\dag{\bm{\alpha}}\le P_{R,\max}.\nonumber
\end{align}

Furthermore, according to the basic definition of the convex optimization problem, both the objective function and the constraint conditions are required to be convex. Thus, let $\bm{\beta}=\frac{\bm{a}}{v}$, $b^2=\frac{1}{v^2}$ and $c^2=u^2b^2$, so that there is no fraction in the optimization problem. Depending on the previous process, \eqref{case2:14} is transformed as
\begin{align}
\label{case2:17}
\underset{\bm{\beta},b^2,c^2}{\max}\ \
&{{\bm{\beta}}^\dag{{\bf{h}}_1}{{\bf{h}}_1}^\dag{\bm{\beta}}} \nonumber \\
\text{s.t.} \ \ &{\rho_0}{c^2}+{\rho_1}{\bm{\beta}}^\dag{{\bf{H}}_1}{{\bf{H}}_1}^\dag{\bm{\beta}}+{\rho_2}{b^2}\le1\nonumber\\
&{\bm{\beta}}^\dag{{\bf{h}}_1}{{\bf{h}}_1}^\dag{\bm{\beta}}\ge \frac{{\rho_0} c^2}{\varepsilon}\nonumber\\
&{\bm{\beta}}^\dag{{\bf{h}}_2}{{\bf{h}}_2}^\dag{\bm{\beta}}\le{c^2}\\
&{\bm{\beta}}^\dag{{\bf{h}}_3}{{\bf{h}}_3}^\dag{\bm{\beta}}{P_S}+{\bm{\beta}}^\dag{{\bf{h}}_4}{{\bf{h}}_4}^\dag{\bm{\beta}}{P_J}+{\bm{\beta}}^\dag{{\bf{H}}_4}{{\bf{H}}_4}^\dag{\bm{\beta}}\ge {Q}{b^2} \nonumber \\
&{\bm{\beta}}^\dag{{\bf{H}}_2}{{\bf{H}}_2}^\dag{\bm{\beta}}{P_S}+{\bm{\beta}}^\dag{{\bf{H}}_3}{{\bf{H}}_3}^\dag{\bm{\beta}}{P_J}+{\bm{\beta}}^\dag{\bm{\beta}}\le {P_{R,\max}}{b^2}.\nonumber
\end{align}

Observing \eqref{case2:17}, its form is similar to SDP problem. Thus, let $\mathbf{X}=\bm{\beta}\bm{\beta}^\dag$, so \eqref{case2:17} is equivalent to the following optimal problem
\begin{align}
\label{case2:18}
\underset{\mathbf{X},b^2,c^2}{\max}\ \
&tr({{{\bf{h}}_1}{{\bf{h}}_1}^\dag \mathbf{X}}) \nonumber \\
\text{s.t.} \ \ &{\rho_0}{c^2}+tr({\rho_1}{{\bf{H}}_1}{{\bf{H}}_1}^\dag \mathbf{X})+{\rho_2}{b^2}\le1\nonumber\\
&tr({{\bf{h}}_2}{{\bf{h}}_2}^\dag \mathbf{X})\le{c^2}, tr({{\bf{h}}_1}{{\bf{h}}_1}^\dag \mathbf{X})\ge \frac{{\rho_0} c^2}{\varepsilon}\\
&tr[({{\bf{h}}_3}{{\bf{h}}_3}^\dag{P_S}+{{\bf{h}}_4}{{\bf{h}}_4}^\dag{P_J}+{{\bf{H}}_4}{{\bf{H}}_4}^\dag)\mathbf{X}]\ge {Q}{b^2} \nonumber \\
&tr[({{\bf{H}}_2}{{\bf{H}}_2}^\dag{P_S}+{{\bf{H}}_3}{{\bf{H}}_3}^\dag{P_J}+{{\bf{I}}})\mathbf{X}]\le {P_{R,\max}}{b^2}\nonumber\\
&\mathbf{X}\succeq 0,rank(\mathbf{X})\le1.\nonumber
\end{align}

Apparently, $rank(\mathbf{X})\le1$ is non-convex. To eliminate this constraint, the semi-determined relaxation method is employed. Further, \eqref{case2:18} is changed to the following standard convex problem
\begin{align}
\label{case2:19}
\underset{\mathbf{X},b^2,c^2}{\max}\ \
&tr({{{\bf{h}}_1}{{\bf{h}}_1}^\dag \mathbf{X}}) \nonumber \\
\text{s.t.} \ \ &{\rho_0}{c^2}+tr({\rho_1}{{\bf{H}}_1}{{\bf{H}}_1}^\dag \mathbf{X})+{\rho_2}{b^2}\le1, \mathbf{X}\succeq 0 \nonumber\\
&tr({{\bf{h}}_2}{{\bf{h}}_2}^\dag \mathbf{X})\le{c^2},tr({{\bf{h}}_1}{{\bf{h}}_1}^\dag \mathbf{X})\ge \frac{{\rho_0} c^2}{\varepsilon}\\
&tr[({{\bf{h}}_3}{{\bf{h}}_3}^\dag{P_S}+{{\bf{h}}_4}{{\bf{h}}_4}^\dag{P_J}+{{\bf{H}}_4}{{\bf{H}}_4}^\dag)\mathbf{X}]\ge {Q}{b^2} \nonumber \\
&tr[({{\bf{H}}_2}{{\bf{H}}_2}^\dag{P_S}+{{\bf{H}}_3}{{\bf{H}}_3}^\dag{P_J}+{{\bf{I}}})\mathbf{X}]\le {P_{R,\max}}{b^2}.\nonumber
\end{align}

The problem \eqref{case2:19} can be solved in the polynomial time using the standard interior point method. Besides, the optimal solution in \eqref{case2:19} may not meet the rank-1 constraint, which means, it might not be the optimal solution in \eqref{case2:18}. The following theorem addresses this challenge. Therefore, there is a conclusion that the optimal solution of the original optimization problem does not change during this process.

\begin{thm}
Assuming that $(\mathbf{X},b,c)$ is the optimal solution for optimization problem \eqref{case2:19}, then there is always a $(\mathbf{X}^\star,b^\star,c^\star)$ satisfying $rank(\mathbf{X}^\star)=1$ and we can find it in polynomial time.
\end{thm}
\begin{proof}
The proof of this theorem is relegated to appendix B.
\end{proof}

Likewise, \eqref{case2:28} and \eqref{case2:30} can be converted to relevant standard convex SDP problems, which are similar to \eqref{case2:19}. Their specific forms can be shown as
\begin{align}
\label{case2:70}
\underset{\mathbf{X},b^2,c^2}{\max}\ \
&tr({{\breve{\bf{h}}_1}{\breve{\bf{h}}_1}^\dag \mathbf{X}}) \nonumber \\
\text{s.t.} \ \ &{\rho_0}{c^2}+tr({\rho_1}{\breve{\bf{H}}_3}{\breve{\bf{H}}_3}^\dag \mathbf{X})+{\rho_2}{b^2}\le1, \mathbf{X}\succeq 0 \nonumber\\
&tr({\breve{\bf{h}}_2}{\breve{\bf{h}}_2}^\dag \mathbf{X})\le{c^2},tr(\breve{{\bf{h}}_1}{\breve{\bf{h}}_1}^\dag \mathbf{X})\ge \frac{{\rho_0} c^2}{\varepsilon}\\
&tr[({\breve{\bf{h}}_3}{\breve{\bf{h}}_3}^\dag{P_S}+{\breve{\bf{h}}_4}{\breve{\bf{h}}_4}^\dag{P_J}+{\breve{\bf{H}}_4}{\breve{\bf{H}}_4}^\dag)\mathbf{X}]\ge {Q}{b^2} \nonumber \\
&tr({\breve{\mathbf{\Phi}}} \mathbf{X})\le {P_{R,\max}}{b^2}.\nonumber
\end{align}
where, ${\breve{\bf{h}}_1}=({\hat{\bf{h}}_1},0)^T$, ${\breve{\bf{h}}_2}=({\hat{\bf{h}}_2},0)^T$, ${\breve{\bf{h}}_3}=({\hat{\bf{h}}_3},0)^T$, ${\breve{\bf{h}}_4}=({\hat{\bf{h}}_4},0)^T$, ${{\breve{\bf{H}}_3}}= \begin{bmatrix}{{\hat{\bf{H}}_3}} & \bf{0}\\ \bf{0} & {{\hat{\bf{H}}_3}} \end{bmatrix}$, ${{\breve{\bf{H}}_4}}= \begin{bmatrix}{{\hat{\bf{H}}_4}} & \bf{0}\\ \bf{0} & {{\hat{\bf{H}}_4}} \end{bmatrix}$, ${\breve{\mathbf{\Phi}}}= \begin{bmatrix}{\hat{\mathbf{\Phi}}} & 0\\ 0 & \bf{I} \end{bmatrix}$. Besides, the optimal problem of \eqref{case2:30} can be represented as
\begin{align}
\label{case2:71}
\underset{\mathbf{X},b^2,c^2}{\max}\ \
&tr({{\acute{\bf{h}}_1}{\acute{\bf{h}}_1}^\dag \mathbf{X}}) \nonumber \\
\text{s.t.} \ \ &{\rho_0}{c^2}+tr({\rho_1}{\acute{\bf{H}}_1}{\acute{\bf{H}}_1}^\dag \mathbf{X})+{\rho_2}{b^2}\le1, \mathbf{X}\succeq 0 \nonumber\\
&tr({\acute{\bf{h}}_2}{\acute{\bf{h}}_2}^\dag \mathbf{X})\le{c^2},tr(\acute{{\bf{h}}_1}{\acute{\bf{h}}_1}^\dag \mathbf{X})\ge \frac{{\rho_0} c^2}{\varepsilon}\\
&tr[({\acute{\bf{h}}_3}{\acute{\bf{h}}_3}^\dag{P_S}+{\acute{\bf{h}}_4}{\acute{\bf{h}}_4}^\dag{P_J}+{\acute{\bf{H}}_4}{\acute{\bf{H}}_4}^\dag)\mathbf{X}]\ge {Q}{b^2} \nonumber \\
&tr[({\acute{\bf{H}}_2}{\acute{\bf{H}}_2}^\dag{P_S}+{\acute{\bf{H}}_3}{\acute{\bf{H}}_3}^\dag{P_J}+{\acute{\bf{I}}})\mathbf{X}]\le {P_{R,\max}}{b^2}.\nonumber
\end{align}
where, ${\acute{\bf{h}}_1}=({\tilde{\bf{h}}_1},0)^T$, ${\acute{\bf{h}}_2}=({\tilde{\bf{h}}_2},0)^T$, ${\acute{\bf{h}}_3}=({\tilde{\bf{h}}_3},0)^T$, ${\acute{\bf{h}}_4}=({\tilde{\bf{h}}_4},0)^T$, ${{\acute{\bf{H}}_1}}= \begin{bmatrix}{{\tilde{\bf{H}}_1}} & \bf{0}\\ \bf{0} & {{\tilde{\bf{H}}_1}} \end{bmatrix}$, ${{\acute{\bf{H}}_2}}= \begin{bmatrix}{{\tilde{\bf{H}}_2}} & \bf{0}\\ \bf{0} & \bf{0} \end{bmatrix}$,${{\acute{\bf{H}}_3}}= \begin{bmatrix}{{\tilde{\bf{H}}_3}} & \bf{0}\\ \bf{0} & \bf{0} \end{bmatrix}$, ${{\acute{\bf{H}}_4}}= \begin{bmatrix}{{\tilde{\bf{H}}_4}} & \bf{0}\\ \bf{0} & {{\tilde{\bf{H}}_4}} \end{bmatrix}$, ${{\acute{\bf{I}}}}= \begin{bmatrix}{{\tilde{\bf{I}}}} & \bf{0}\\ \bf{0} & {{\tilde{\bf{I}}}} \end{bmatrix}$.

\subsection{Achievable Rate-Energy Region with a jammer}
Since that \eqref{case2:19}, \eqref{case2:70} and \eqref{case2:71} are equivalent, only \eqref{case2:19} is used here as an example. The research on the achievable rate-energy region is available to analyze the tradeoffs in the transmission of signal and energy with the EH and sum power constraints \cite{t12}. In this scenario, the achievable rate-energy region can be defined as
\begin{equation}
\label{case2:58}
\mathcal{C}\triangleq \left\{\begin{aligned}(R, Q): &0\le R\le \frac{1}{2}\log(1+\text{SINR}), 0\le Q\le Q_{max},\\
&\mathbb{E}[|{\mathbf{x} _R}{|^2}]\le{P_{R, \max }}
\end{aligned}\right\}.
\end{equation}

To characterize the rate-energy region, we examine three boundary points, denoted by $(R_{\text{EH}}, Q_{\max})$, $(0, Q_{\max})$ and $(R_{\max}, 0)$, where $R_{\max}$ is the maximum allowable relay transmission rate regardless of the work state of the EH receiver, $ Q_{\max}$ refers to the maximum harvesting energy that the EH node can be obtained regardless of the transmission rate and $R_{\text{EH}}$ represents the maximum transmission rate when $Q=Q_{\max}$.

Calculating $R_{\max}$, by definition, the EH constraint can be removed, and we have that
\begin{align}
\label{case2:59}
\underset{\bm{\alpha}}{\max}\ \
&\frac{{\bm{\alpha}}^\dag{{\bf{h}}_1}{{\bf{h}}_1}^\dag{\bm{\alpha}}{P_S}}{{\bm{\alpha}}^\dag{{\bf{h}}_2}{{\bf{h}}_2}^\dag{\bm{\alpha}}{P_J}+{\bm{\alpha}}^\dag{{\bf{H}}_1}{{\bf{H}}_1}^\dag{\bm{\alpha}}{{\sigma}_R}^2+{{{\sigma}_D}^2}}\nonumber \\
\text{s.t.} \ \ &{\bm{\alpha}}^\dag \mathbf{\Sigma}{\bm{\alpha}}\le P_{R,\max}\nonumber\\
&\varepsilon{\bm{\alpha}}^\dag{{\bf{h}}_1}{{\bf{h}}_1}^\dag{\bm{\alpha}}{P_S}\ge {\bm{\alpha}}^\dag{{\bf{h}}_2}{{\bf{h}}_2}^\dag{\bm{\alpha}}{P_J}.
\end{align}

Similar to the measure used in subsection $A$, let $\mathbf{X}=\bm{\alpha}\bm{\alpha}^\dag$, \eqref{case2:59} can be converted to
\begin{align}
\label{case2:61}
\underset{\mathbf{X},b^2,c^2}{\max}\ \
&tr({{{\bf{h}}_1}{{\bf{h}}_1}^\dag \mathbf{X}}) \nonumber \\
\text{s.t.} \ \ &{\rho_0}{c^2}+tr({\rho_1}{{\bf{H}}_1}{{\bf{H}}_1}^\dag \mathbf{X})+{\rho_2}{b^2}\le1,tr({{\bf{h}}_2}{{\bf{h}}_2}^\dag \mathbf{X})\le{c^2}\nonumber\\
&tr[({{\bf{H}}_2}{{\bf{H}}_2}^\dag{P_S}+{{\bf{H}}_3}{{\bf{H}}_3}^\dag{P_J}+{\sigma_R}^2{{\bf{I}}})\mathbf{X}]\le {P_{R,\max}}{b^2}\nonumber\\
&tr({{\bf{h}}_1}{{\bf{h}}_1}^\dag \mathbf{X})\ge \frac{{\rho_0} c^2}{\varepsilon},rank(\mathbf{X})\le1.
\end{align}
where, ${{\bf{h}}_1}={{\bf{h}}_S}^\ast\otimes {{\bf{h}}_D}$, ${{\bf{h}}_2}={{\bf{h}}_J}^\ast\otimes {{\bf{h}}_D}$,  ${{\bf{H}}_1}={{\bf{I}}}\otimes {{\bf{h}}_D}$, ${{\bf{H}}_2}={{\bf{h}}_S}^\ast\otimes {{\bf{I}}}$, ${{\bf{H}}_3}={{\bf{h}}_J}^\ast\otimes {{\bf{I}}}$, and when $\varepsilon \in [0,1)$ approaches 1, it does not affect global optimality. Clearly, as the proof procedure of subsection $A$, the rank-1 constraint can be relaxed and \eqref{case2:61} is equivalent to the original problem.

Regardless of the transmission rate, $Q_{\max}$ can be obtained by
\begin{align}
\label{case2:62}
\underset{\bm{\alpha}}{\max}\ \
&{\bm{\alpha}}^\dag{{\bf{h}}_3}^\ast{{\bf{h}}_3}^T{\bm{\alpha}}{P_S}+{\bm{\alpha}}^\dag{{\bf{h}}_4}^\ast{{\bf{h}}_4}^T{\bm{\alpha}}{P_J}+{\bm{\alpha}}^\dag{{\bf{H}}_4}^\dag{{\bf{H}}_4}{\bm{\alpha}}{\sigma_R}^2\nonumber\\
\text{s.t.} \ \ & {\bm{\alpha}}^\dag \mathbf{\Sigma} {\bm{\alpha}} \le P_{R,\max}.
\end{align}

Meanwhile, let $\mathbf{X}=\bm{\alpha}\bm{\alpha}^\dag$, and we have the following optimal problem transformed from \eqref{case2:62}.
\begin{align}
\label{case2:63}
\underset{\mathbf{X}}{\max}\ \
&tr[({{\bf{h}}_3}{{\bf{h}}_3}^\dag{P_S}+{{\bf{h}}_4}{{\bf{h}}_4}^\dag{P_J}+{{\bf{H}}_4}{{\bf{H}}_4}^\dag{\sigma_R}^2)\mathbf{X}] \nonumber \\
\text{s.t.} \ \
&tr(\mathbf{\Sigma}\mathbf{X})\le {P_{R,\max}}\nonumber \\
&rank(\mathbf{X})\le1\nonumber \\
&\mathbf{X}\succeq 0.
\end{align}

According to the literature \cite{t13}, the rank-1 constraint can be relaxed. Therefore, the problem is an SDP problem.

For the optimal problem of point $(R_{\text{EH}}, Q_{\max})$, in the light of the preceding analysis, it can be provided as
\begin{align}
\label{case2:64}
\underset{\bm{\alpha}}{\max}\ \
&\frac{{\bm{\alpha}}^\dag{{\bf{h}}_1}{{\bf{h}}_1}^\dag{\bm{\alpha}}{P_S}}{{\bm{\alpha}}^\dag{{\bf{h}}_2}{{\bf{h}}_2}^\dag{\bm{\alpha}}{P_J}+{\bm{\alpha}}^\dag{{\bf{H}}_1}{{\bf{H}}_1}^\dag{\bm{\alpha}}{{\sigma}_R}^2+{{{\sigma}_D}^2}}\nonumber \\
\text{s.t.} \ \ &{\bm{\alpha}}^\dag \mathbf{\Sigma}{\bm{\alpha}}\le P_{R,\max}\nonumber\\
&\varepsilon{\bm{\alpha}}^\dag{{\bf{h}}_1}{{\bf{h}}_1}^\dag{\bm{\alpha}}{P_S}\ge {\bm{\alpha}}^\dag{{\bf{h}}_2}{{\bf{h}}_2}^\dag{\bm{\alpha}}{P_J}\nonumber \\
&{\bm{\alpha}}^\dag{{\bf{h}}_3}{{\bf{h}}_3}^\dag{\bm{\alpha}}{P_S}+{\bm{\alpha}}^\dag{{\bf{h}}_4}{{\bf{h}}_4}^\dag{\bm{\alpha}}{P_J}+{\bm{\alpha}}^\dag{{\bf{H}}_4}{{\bf{H}}_4}^\dag{\bm{\alpha}}{{\sigma}_R}^2\ge Q_{\max}.
\end{align}

In view of the work in subsection $A$, \eqref{case2:64} can be transformed into the following SDP problem
\begin{align}
\label{case2:65}
\underset{\mathbf{X},b^2,c^2}{\max}\ \
&tr({{{\bf{h}}_1}{{\bf{h}}_1}^\dag \mathbf{X}}) \nonumber \\
\text{s.t.} \ \ &{\rho_0}{c^2}+tr({\rho_1}{{\bf{H}}_1}{{\bf{H}}_1}^\dag \mathbf{X})+{\rho_2}{b^2}\le1\nonumber\\
&tr({{\bf{h}}_2}{{\bf{h}}_2}^\dag \mathbf{X})\le{c^2},tr({{\bf{h}}_1}{{\bf{h}}_1}^\dag \mathbf{X})\ge \frac{{\rho_0} c^2}{\varepsilon}\\
&tr[({{\bf{h}}_3}{{\bf{h}}_3}^\dag{P_S}+{{\bf{h}}_4}{{\bf{h}}_4}^\dag{P_J}+{{\bf{H}}_4}{{\bf{H}}_4}^\dag)\mathbf{X}]\ge {Q_{\max}}{b^2} \nonumber \\
&tr(\mathbf{\Sigma} \mathbf{X})\le {P_{R,\max}}{b^2},\mathbf{X}\succeq 0.\nonumber
\end{align}

\subsection{Closed-form expression and suboptimal solution}
In this subsection, in order to analyze the rate-energy region, the closed-form and suboptimal solution of the two points, $(R_{\max}, 0)$ and $(0, Q_{\max})$, should be taken into consideration accordingly.

\subsubsection{$(R_{\max}, 0)$}
We consider the optimal problem \eqref{case2:59} over $R_{\max}$ at first. Since the second constraint in \eqref{case2:59} should be addressed as ratio to ensure the lower error, we discuss two cases, ${{\bf{h}}_S} \parallel {{\bf{h}}_J}$ and ${{\bf{h}}_S} \nparallel {{\bf{h}}_J}$.
\paragraph{${{\bf{h}}_S} \parallel {{\bf{h}}_J}$}
Suppose that $\mathbf{h}_J=\rho \mathbf{h}_S$, then, there is $\mathbf{h}_2=\rho \mathbf{h}_1$. Because $\frac{\mathbb{E}[|{x_{J,eq}}{|^2}]}{\mathbb{E}[|{x_{S,eq}}{|^2}]}$ is established as
\begin{align}
\label{case2:31}
&\frac{\mathbb{E}[|{x_{J,eq}}{|^2}]}{\mathbb{E}[|{x_{S,eq}}{|^2}]}=\frac{|{{\bf{h}}_D}^\dag{\bf{A}}{{\bf{h}}_J}{|^2}{P_J}}{|{{\bf{h}}_D}^\dag{\bf{A}}{{\bf{h}}_S}{|^2}{P_S}}=\frac{|{{\bf{h}}_D}^\dag{\bf{A}}{{\bf{h}}_S}{|^2}|{\rho}{^2}{P_J}}{|{{\bf{h}}_D}^\dag{\bf{A}}{{\bf{h}}_S}{|^2}{P_S}}=\frac{|{\rho}{|^2}P_J}{P_S}.
\end{align}
considering \eqref{capacity}, it is easy to know that when $\frac{|{\rho}{|^2}P_J}{P_S}\ge1$, the deterministic coding capacity of the equivalent GAVC(A) is zero.

For the same reason, when $\frac{|{\rho}{|^2}P_J}{P_S}<1$, the second constraint in \eqref{case2:59} holds. Hence, for the basic definition of convex problem, the objective function of \eqref{case2:59} which is  non-convex should be investigated. From the physical point of view, if ${\bm{\alpha}}^\dag \mathbf{\Sigma} {\bm{\alpha}} = P_{R,\max}$ is held and the energy harvesting node is not working, the maximum value of the objective function is acquired.

Substituting $\frac{{\bm{\alpha}}^\dag \mathbf{\Sigma} {\bm{\alpha}}}{P_{R,\max}} =1$ and $\mathbf{h}_2=\rho \mathbf{h}_1$ into optimization problem \eqref{case2:59}, it can be simplified as
\begin{align}
\label{case2:32}
\underset{\bm{\alpha}}{\max}\ \
&\frac{{\bm{\alpha}}^\dag{{\bf{h}}_1}^\ast{{\bf{h}}_1}^T{\bm{\alpha}}{P_S}}{{\bm{\alpha}}^\dag({{\bf{h}}_1}^\ast{{\bf{h}}_1}^T|{\rho}{|^2}{P_J}+{{\bf{H}}_3}^\dag{{\bf{H}}_3}+\frac{\mathbf{\Sigma}}{P_{R,\max}}){\bm{\alpha}}}.
\end{align}

Given that the matrix ${{\bf{h}}_1}^\ast{{\bf{h}}_1}^T|{\rho}{|^2}{P_J}+{{\bf{H}}_3}^\dag{{\bf{H}}_3}+\frac{ \mathbf{\Sigma}}{P_{R,\max}}$ is a positive definite matrix, it can be known from the existing conclusion that when $f(\bm{a})=\frac{\bm{a}^\dag \mathbf{h}\mathbf{h}^\dag \bm{a}}{\bm{a}^\dag \mathbf{P} \bm{a}}$, the maximum value of the function $f(\bm{a})$ can be expressed as $ \mathbf{h}^\dag \mathbf{P}^{-1}\mathbf{h} $, and the corresponding vector solution can be written as $\bm{a}=\kappa\mathbf{P}^{-1}\mathbf{h}$, where $\kappa$ is an arbitrary constant \cite{c14}. Therefore, under ${{\bf{h}}_S} \parallel {{\bf{h}}_J}$, the maximum value of the objective function of \eqref{case2:59} is as follows
\begin{align}
\label{case2:33}
\mathbf{P}=[{{\bf{h}}_1}^\ast{{\bf{h}}_1}^T|{\rho}{|^2}{P_J}+{{\bf{H}}_3}^\dag{{\bf{H}}_3}+\frac{\mathbf{\Sigma}}{P_{R,\max}}]^{-1}.
\end{align}
\begin{align}
\label{case2:34}
\text{SINR}={{\bf{h}}_1}^T\mathbf{P}{{\bf{h}}_1}^\ast{P_S}.
\end{align}
\begin{align}
\label{case2:35}
{\bm{\alpha}}=\kappa\mathbf{P}{{\bf{h}}_1}^\ast{P_S}.
\end{align}
\begin{align}
\label{case2:36}
&\kappa=e^{j\theta}\sqrt{\frac{{P_{R,\max}}}{\lVert\mathbf{\Sigma}^{\frac{1}{2}}\mathbf{P} \mathbf{h}_1^\ast\lVert^2}}.
\end{align}
where, $\theta$ is an arbitrary angle.

\paragraph{${{\bf{h}}_S} \nparallel {{\bf{h}}_J}$}
Assuming that $\mathbf{h}_S \nparallel \mathbf{h}_J$, then, there is $\mathbf{h}_1 \nparallel \mathbf{h}_2$. Let $\bm{W}$ be a matrix of eigenvectors corresponding to the zero eigenvalues of matrix $\mathbf{h}_J^\ast\mathbf{h}_J^{T}$.  Since ${\bm{W}}$ is not full rank, this solution is not optimal. If $\bm{\alpha}=\bm{W}\bm{f}$ is assumed, we have
\begin{align}
\label{case2:37}
&\bm{\alpha}^\dag\mathbf{h}_2^\ast\mathbf{h}_2^{T}\bm{\alpha}=\bm{f}^\dag\bm{W}^\dag\mathbf{h}_2^\ast\mathbf{h}_2^{T}\bm{W}\bm{f}=0.
\end{align}
\begin{align}
\label{case2:38}
&\bm{\alpha}^\dag\mathbf{h}_1^\ast\mathbf{h}_1^{T}\bm{\alpha}=\bm{f}^\dag\bm{W}^\dag\mathbf{h}_1^\ast\mathbf{h}_1^{T}\bm{W}\bm{f}>0.
\end{align}

Similar to the analysis at ${{\bf{h}}_S} \parallel {{\bf{h}}_J}$, from the physical point of view, substituting $\frac{{\bm{\alpha}}^\dag \mathbf{\Sigma} {\bm{\alpha}}}{P_{R,\max}} =1$ and $\bm{\alpha}=\bm{W}\bm{f}$ into \eqref{case2:59}, it can be transformed into the following form
\begin{align}
\label{case2:39}
\underset{\bm{f}}{\max}\ \ &\frac{\bm{f}^\dag\bm{W}^\dag\mathbf{h}_1^\ast\mathbf{h}_1^T\bm{W}\bm{f}}{\bm{f}^\dag\bm{W}^\dag(\mathbf{H}_3^\dag\mathbf{H}_3+\frac{\mathbf{\Sigma}}{P_{R, \max }})\bm{W}\bm{f}}.
\end{align}

Evidently, the optimization problem can be equivalent to solve the optimal value of the function. Considering that the matrix ${\bm{W}^\dag}(\mathbf{H}_3^\dag\mathbf{H}_3+\frac{\mathbf{\Sigma}}{P_{R, \max }}){\bm{W}}$ is a positive definite matrix, then the solving method is similar to ${{\bf{h}}_S} \parallel {{\bf{h}}_J}$. According to the existing conclusion of the positive definite matrix, the suboptimal solution of the optimization problem is that
\begin{align}
\label{case2:40}
\mathbf{P}=[{\bm{W}^\dag}(\mathbf{H}_3^\dag\mathbf{H}_3+\frac{\mathbf{\Sigma}}{P_{R, \max }}){\bm{W}}]^{-1}.
\end{align}
\begin{align}
\label{case2:41}
{\text{SINR}_{sub}}={{\bf{h}}_1}^T{\bm{W}}\mathbf{P}{\bm{W}^\dag}{{\bf{h}}_1}^\ast.
\end{align}

The corresponding vector of the solution is
\begin{align}
\label{case2:42}
{\bm{f}}=\kappa\mathbf{P}{\bm{W}^\dag}{{\bf{h}}_1}^\ast.
\end{align}
\begin{align}
\label{case2:43}
&\kappa=e^{j\theta}\sqrt{\frac{{P_{R,\max}}}{\lVert\mathbf{\Sigma}^{\frac{1}{2}}\mathbf{P} {\bm{W}^\dag}{{\bf{h}}_1}^\ast\lVert^2}}.
\end{align}
where, $\theta$ is an arbitrary angle.

The suboptimal solution of the original optimization problem is $\bm{W}\bm{f}$. The relative expression of it can be presented as
\begin{align}
\label{case2:44}
&\bm{\alpha}=\bm{W}\bm{f}=\kappa\bm{W}\mathbf{P}{\bm{W}^\dag}{{\mathbf{h}}_1}^\ast.
\end{align}
\begin{align}
\label{case2:45}
&\mathbf{A}={\mathbf{U}_1}^\ast({{vec^{-1}(\bm{b})}})^T{\mathbf{U}_1}^\dag.
\end{align}

\subsubsection{$(0,Q_{\max})$}
Similar to the measure used in $(R_{\max}, 0)$, we again discuss two cases, ${{\bf{h}}_S} \parallel {{\bf{h}}_J}$ and ${{\bf{h}}_S} \nparallel {{\bf{h}}_J}$. Technically, the optimal value of energy harvesting is achievable, if the power of the relay is only used to transform energy. Based on this, the constraint of the transmission rate is not considered. Then, the optimal problem over $Q_{\max}$ can be shown as \eqref{case2:62}. The detailed process will be described in the following.

\paragraph{${{\bf{h}}_S} \parallel {{\bf{h}}_J}$}
Assuming that ${{\bf{h}}_J}={\rho}{{\bf{h}}_S}$, then, ${{\bf{h}}_2}={\rho}{{\bf{h}}_1}$ and ${{\bf{h}}_4}={\rho}{{\bf{h}}_3}$ can be obtained. In fact, if the constraint in \eqref{case2:62} holds, the objective function in \eqref{case2:62} can take the maximum value. Substituting the conditions $\frac{{\bm{\alpha}}^\dag \mathbf{\Sigma} {\bm{\alpha}}}{P_{R,\max}} =1$, ${{\bf{h}}_2}={\rho}{{\bf{h}}_1}$ and $\mathbf{h}_4=\rho\mathbf{h}_3$ into \eqref{case2:62}, thus, there is the following expression
\begin{align}
\label{case2:48}
\underset{\bm{\alpha}}{\max}\ \
&{\bm{\alpha}}^\dag{{\bf{h}}_3}^\ast{{\bf{h}}_3}^T{\bm{\alpha}}{P_S}+{\bm{\alpha}}^\dag{{\bf{h}}_3}^\ast{{\bf{h}}_3}^T{\bm{\alpha}}{P_J}{\rho}^2+{\bm{\alpha}}^\dag{{\bf{H}}_4}^\dag{{\bf{H}}_4}{\bm{\alpha}}{\sigma_R}^2\nonumber\\
\text{s.t.} \ \ & {\bm{\alpha}}^\dag \mathbf{\Sigma} {\bm{\alpha}} = P_{R,\max}.
\end{align}

Solving \eqref{case2:48}, several conclusions according to\cite{t5} can be given
\begin{align}
\label{case2:49}
\boldsymbol{\Gamma}={{\bf{h}}_3}^\ast{{\bf{h}}_3}^T({P_S}+{\rho}^2{P_J})+{{\bf{H}}_4}^\dag{{\bf{H}}_4}{\sigma_R}^2.
\end{align}
\begin{align}
\label{case2:50}
{\bm{\alpha}}= \sqrt{P_{R,\max}}(\lVert\mathbf{\Sigma}^{\frac{1}{2}}\Psi_{max}(\mathbf{\Sigma},\boldsymbol{\Gamma}){\rVert})^{-1}\Psi_{max}(\mathbf{\Sigma},\boldsymbol{\Gamma}).
\end{align}
where, $\Psi_{max}(\bf A,B)$ denotes the generalized eigenvector corresponding to the largest generalized eigenvalue of the matrix pair $(\bf{A},\bf{B})$. Thus, the optimal value $Q$ of the object function can be expressed as
\begin{align}
\label{case2:51}
Q= \lambda_{max}(\mathbf{\Sigma}^{-1}\boldsymbol{\Gamma}) P_{R,\max}.
\end{align}
where, $\lambda_{max}(\bf{D})$ denotes the largest eigenvalue of matrix $\bf{D}$.

\paragraph{${{\bf{h}}_S} \nparallel {{\bf{h}}_J}$}
Let ${{\bf{h}}_S} \nparallel {{\bf{h}}_J}$, then, ${{\bf{h}}_1} \nparallel {{\bf{h}}_2}$ and ${{\bf{h}}_3} \nparallel {{\bf{h}}_4}$ can be gotten. Assuming that ${\bm{W}}$ is a matrix of eigenvectors corresponding to the zero eigenvalues of matrix ${{\bf{h}}_1}^\ast{{\bf{h}}_1}^T$, hence, we have
\begin{align}
\label{case2:52}
&\bm{\alpha}^\dag\mathbf{h}_1^\ast\mathbf{h}_1^{T}\bm{\alpha}=\bm{f}^\dag\bm{W}^\dag\mathbf{h}_1^\ast\mathbf{h}_1^{T}\bm{W}\bm{f}=0
\end{align}
\begin{align}
\label{case2:53}
&\bm{\alpha}^\dag\mathbf{h}_2^\ast\mathbf{h}_2^{T}\bm{\alpha}=\bm{f}^\dag\bm{W}^\dag\mathbf{h}_2^\ast\mathbf{h}_2^{T}\bm{W}\bm{f}>0
\end{align}

From the physical point of view, when ${\bm{\alpha}}^\dag \mathbf{\Phi} {\bm{\alpha}} = P_{R,\max}$, considering ${\bm{W}}$ is not full rank, the objective function in \eqref{case2:62} takes the suboptimal solution.

Substituting $\frac{{\bm{\alpha}}^\dag \mathbf{\Sigma} {\bm{\alpha}}}{P_{R,\max}} =1$ and $\bm{\alpha}=\bm{W}\bm{f}$ into \eqref{case2:62}, then the optimization problem can be formulated into the following form
\begin{align}
\label{case2:54}
\underset{\bm{f}}{\max}\ \
&{\bm{f}}^\dag{\bm{W}}^\dag[{{\bf{h}}_3}^\ast{{\bf{h}}_3}^T{P_S}+{{\bf{h}}_4}^\ast{{\bf{h}}_4}^T{P_J}+{{\bf{H}}_4}^\dag{{\bf{H}}_4}{\sigma_R}^2]{\bm{W}}{\bm{f}}\nonumber\\
\text{s.t.} \ \ & {\bm{f}}^\dag{\bm{W}}^\dag \mathbf{\Sigma} {\bm{W}}{\bm{f}} = P_{R,\max}
\end{align}

By analyzing \eqref{case2:54}, the following conclusions can be drawn
\begin{align}
\label{case2:55}
\mathbf{T}={{\bf{h}}_3}^\ast{{\bf{h}}_3}^T{P_S}+{{\bf{h}}_4}^\ast{{\bf{h}}_4}^T{P_J}+{{\bf{H}}_4}^\dag{{\bf{H}}_4}{\sigma_R}^2
\end{align}
\begin{align}
\label{case2:56}
{\bm{\alpha}}= \sqrt{P_{R,\max}}(\lVert \mathbf{\Sigma}^{\frac{1}{2}}\Psi_{max}(\mathbf{\Sigma},\mathbf{T}){\rVert})^{-1}\Psi_{max}(\mathbf{\Sigma},\mathbf{T})
\end{align}
where, $\Psi_{max}(\bf A,B)$ denotes the generalized eigenvector corresponding to the largest generalized eigenvalue of the matrix pair $(\bf{A},\bf{B})$. Thus, the optimal value $Q$ of the object function can be exhibited as
\begin{align}
\label{case2:57}
Q= \lambda_{max}(\mathbf{\Sigma}^{-1}\mathbf{T}) P_{R,\max}
\end{align}
where $\lambda_{max}(\bf{D})$ denotes the largest eigenvalue of matrix $\mathbf{D}$.

\section{Numerical Results}\label{simulations}
In this section, we illustrate several numerical results to demonstrate the effectiveness of our proposed scheme. All simulations are performed in MATLAB R2015b. We use CVX toolbox \cite{t14} to solve the SDP problems. Suppose that channel coefficients ${{\bf{h}}_S}$, ${{\bf{h}}_J}$ , ${{\bf{h}}_D}$ and ${\mathbf{h}_E}$ are generated by independent complex Gaussian random variables with the distribution $\mathcal{CN}\sim\left( {0,1} \right)$. The variances of noise are $\sigma_R^2=\sigma_D^2=1$. The transmission power at jammer is $P_J=15\text{dBW}$ in Fig. \ref{figure1} and the relay has a power budget $P_{R_{\max}}$. In all simulations, we set $\varepsilon  = 0.99$, and the number of channel samples is set to 1000 so as to analyze the average performance of the proposed schemes.

In Fig. \ref{figure2}, we compare the anti-jamming performance of some existing schemes and the scheme proposed in this paper with the sum power constraint. 1) No jammer, we calculate the optimal beamforming matrix in EHWSN without jammer; 2) Pseudo matched forwarding (PMF), the beamforming matrix is chosen as $\mathbf{A}=\mu{\mathbf{h}_{D}}^\ast{\mathbf{h}_{S}}^\dag$, where $\mu=\sqrt{\frac{P_{R,max}}{\lVert{\mathbf{h}_{D}}\rVert^2(\lVert{\mathbf{h}_{S}}\rVert^4 P_S+|{\mathbf{h}_{S}}^\dag{\mathbf{h}_J}{|}^2 P_J+\lVert{\mathbf{h}_{S}}\rVert^2)}}$; 3) Zero-forcing (ZF) beamforming, the beamforming matrix is generated by $\mathbf{A}=\tau{\mathbf{H}_{\bot}}$, where $\tau=\sqrt{\frac{P_{R,max}}{K(|{\mathbf{h}_{\bot}}{\mathbf{h}_S}{|}^2+1)}}$ and ${\mathbf{H}_{\bot}}$ is a matrix in which each row is ${\mathbf{h}_{\bot}}$. ${\mathbf{h}_{\bot}}$ is randomly selected from the null space of ${\rm span}({\bf{h}}_J)$; 4) Direct relaying (DR), the beamforming matrix is given by $\mathbf{A}=\xi{\mathbf{I}_K}$, where $\xi=\sqrt{\frac{P_{R,max}}{\lVert{\mathbf{h}_{S}}\rVert^2 P_S+\lVert{\mathbf{h}_{J}}\rVert^2 P_J+K}}$. Clearly, although the zero-forcing scheme has good anti-jamming ability, it is still in a disadvantaged situation compared with the scheme proposed in this paper. Other schemes gradually tend to zero as the jamming source power increases. For this reason, the anti-jamming performance of the proposed scheme is better than others.
\begin{figure}[!t]
	\centering
	\includegraphics[width=2.5in]{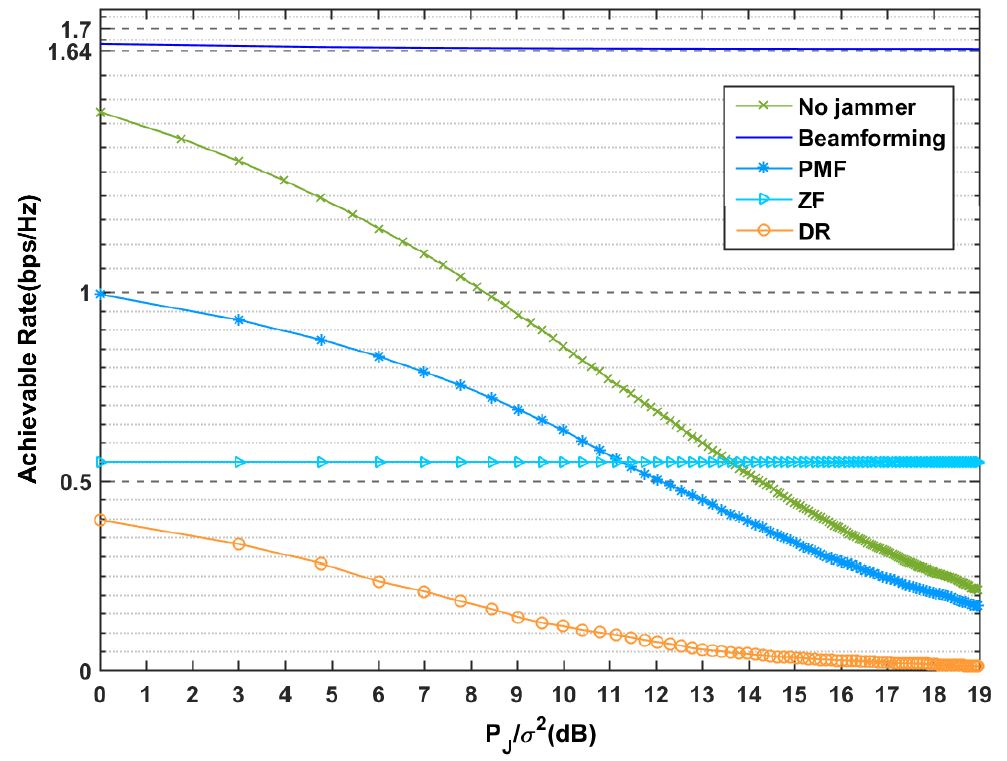}
	\caption{Comparison of the achievable rate with $P_J$ under different schemes.}
    \label{figure2}
\end{figure}
\begin{figure}[htbp]
\begin{minipage}[t]{0.5\linewidth}
\centering
	\includegraphics[width=2.5in]{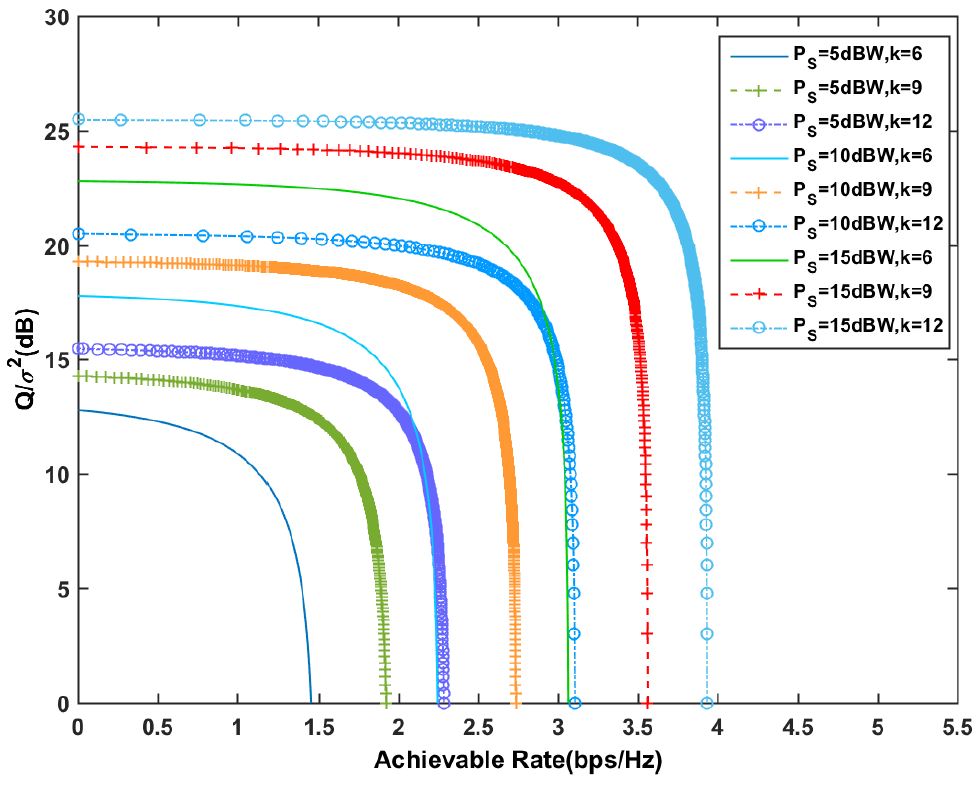}
	\caption{Achievable rate-energy regions for the AF relay network.}
\label{figure1}
\end{minipage}%
\begin{minipage}[t]{0.5\linewidth}
\centering
	\includegraphics[width=2.5in]{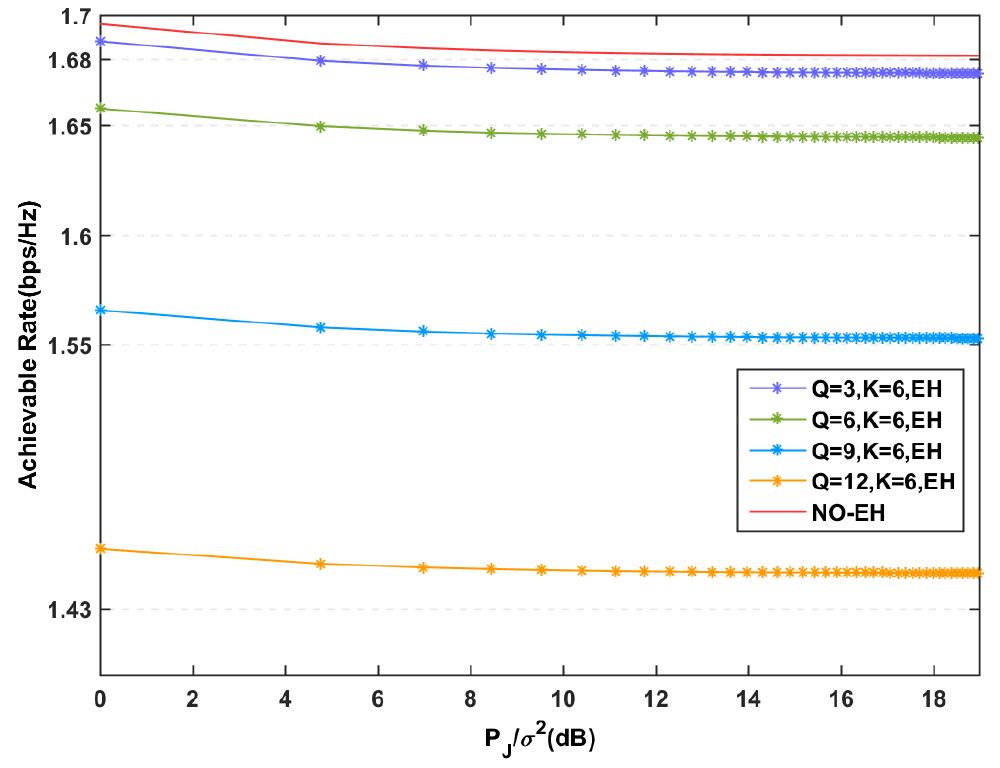}
	\caption{Achievable rate with the increasing $P_J$ under different Q.}
\label{figure3}
\end{minipage}
\end{figure}
\begin{figure}[htbp]
\begin{minipage}[t]{0.5\linewidth}
\centering
	\includegraphics[width=2.5in]{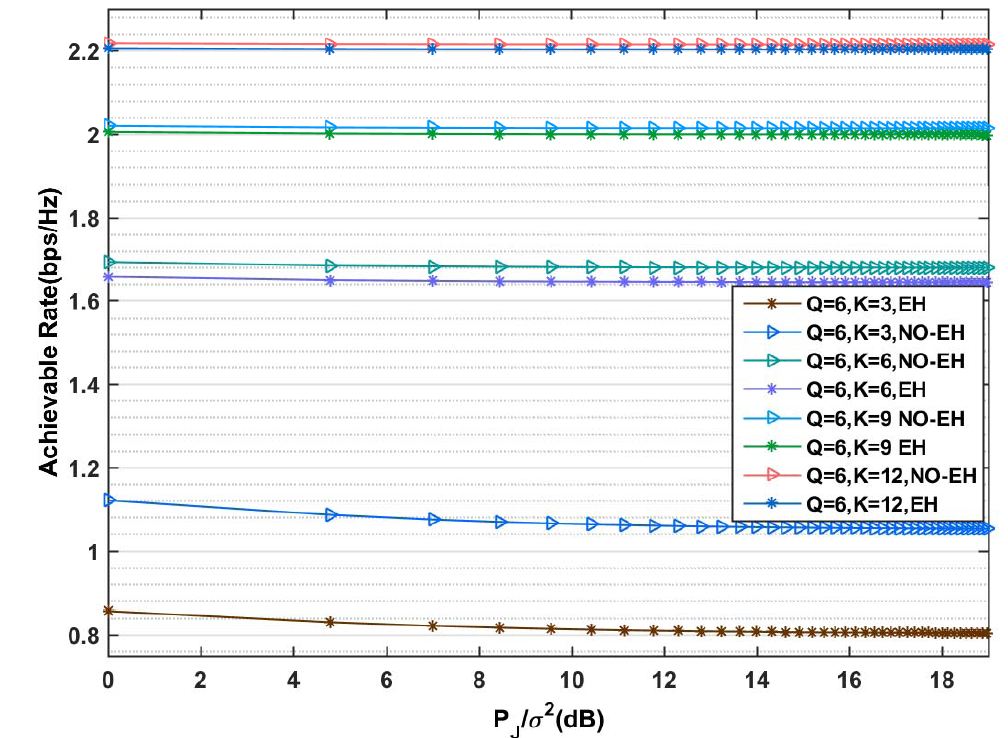}
	\caption{Achievable rate with the increasing $P_J$ under different $K$.}
    \label{figure4}
\end{minipage}%
\begin{minipage}[t]{0.5\linewidth}
\centering
	\includegraphics[width=2.5in]{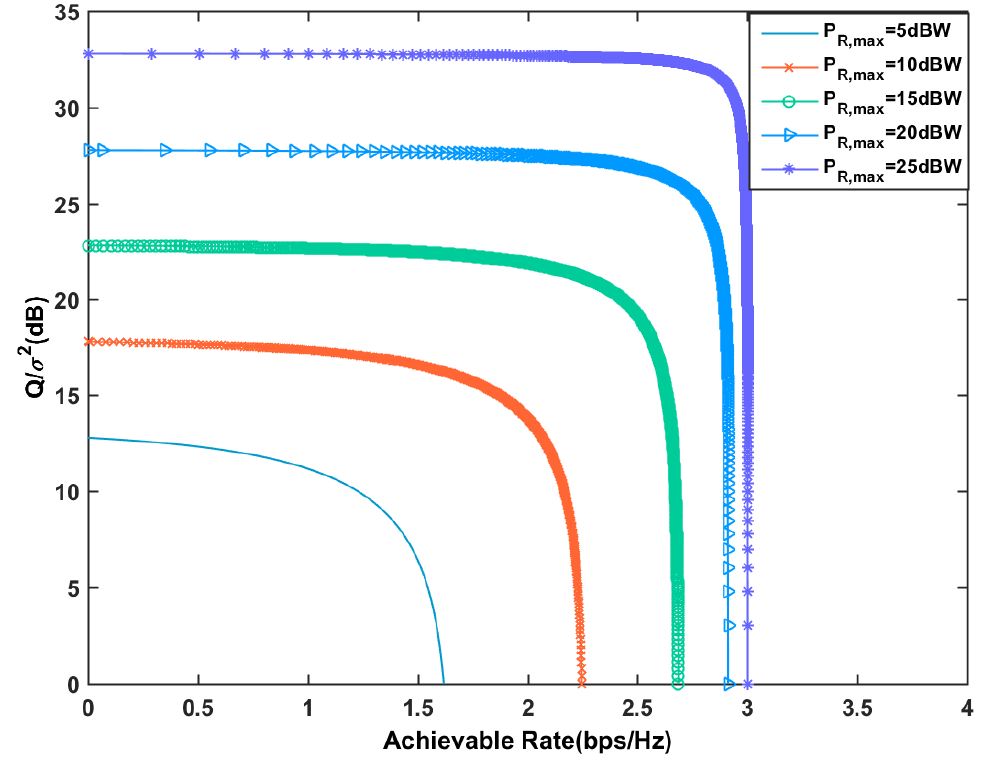}
	\caption{Achievable rate-energy regions for the increasing $P_{R,\max}$.}
    \label{figure5}
\end{minipage}
\end{figure}

The trend for the transmission signal and energy is exhibited in Fig. \ref{figure1}. It can be seen that the rate-energy region is expanding with the number of antennas at the relay. Moreover, given the number of antennas at the relay, the rate-energy region is expanding when the transmission power $P_S$ is increases.

As shown in Fig. \ref{figure3}, given $P_S$ and $K$, we discuss the effect of $Q$ for achievable rate with the increasing $P_J$. Obviously, when the value of $Q$ is climbing, the distance between the model with EH and NO-EH is growing. This phenomenon occurs due to that the power at the relay transferred to the ID receiver is decreasing when the required energy is rising. For the same reason, the achievable rate of the NO-EH model is higher than the EH model.

Fig. \ref{figure4} demonstrates that the achievable rate increases with the relaxation of the jamming power threshold $P_J$ with different $K$. From the plot, it can be seen that the achievable rate is significantly reduced if the jamming power is small because this kind of jamming signal can be seen as noise. However, when the value of $P_J$ is greater than a certain value, the achievable rate remains constant since beamforming technology eliminates jamming signals, making the transmission rate less affected by jamming signals. Furthermore, the achievable rate improves with the increasing antenna number $K$. For the addition of $K$, the spatial diversity is added, and additional diversity gain is obtained, thus it improves the system transmission performance.

Fig. \ref{figure5} shows the change of rate-energy for different relay power budget $P_{R,\max}$. Only the stable values of $P_S$, $P_J$ and $K$ are considered, i.e., $P_S=P_J=6\text{dB}$, $K=6$. We observe from Fig. \ref{figure5} that the rate-energy is expanding when $P_{R,\max}$ is increases, while if $P_{R,\max}$ reaches a certain value, the expansion rate tends to zero.

\section{Conclusion}\label{sec:conclusion}

In this paper, we investigate the optimization of multi-antenna relay with a jammer based on an EHWSRN architecture. To specific calculating process, we model this network as an equivalent GAVC related to the beamforming matrix at the relay. Based on the previous hypothesis, the optimization problem is formulated to maximize the SINR performance under the constraints of EH, sum power and jammer. After further analysis, we showed that this problem is non-convex which is difficult to solve in polynomial time. There are several strategies from three aspects, which are system stability, lower computational complexity and the combination, to change it from non-convex to SDP. The rank-1 constraint of the optimization problem is proved. Simulation results are used to demonstrate the efficiency of the proposed strategy in terms of the achievable rate performance at anti-jamming and energy consumption.

% or


\appendices
\section{The proof of the theorem 1}
Assuming that the fading vectors of the channel  are ${{\bf{h}}_S},{{\bf{h}}_J},{{\bf{h}}_D},{{\bf{h}}_E}$. The singular value decomposition of the constructed matrix $[{{\bf{h}}_S},{{\bf{h}}_J},{{\bf{h}}_D},{{\bf{h}}_E}]$ can be shown as $[{{\bf{h}}_S},{{\bf{h}}_J},{{\bf{h}}_D},{{\bf{h}}_E}]={\bf{U}}{\bf{\Omega}}{\bf{V}}^\dag$, where $\mathbf{U}$ is a unitary matrix and ${\bf{\Omega}}=diag{(\omega_1,\omega_2,...,\omega_r,0,...,0)}$. $\omega_i (i=1,2,...,r)$ are the positive singular values of the matrix $[{{\bf{h}}_S},{{\bf{h}}_J},{{\bf{h}}_D},{{\bf{h}}_E}]$, where $r$ is the number of non-zero singular values. Let $\mathbf{U}=[{\bf{U}}_1,{\bf{U}}_2]$, where ${\bf{U}}_1$ and ${\bf{U}}_2$ represent the former $r$ and the latter $K-r$ columns of ${\bf{U}}$, respectively. The beamforming matrix $\mathbf{A}$ of the optimization problem \eqref{case2:11} can be written as the following form
\begin{align}
\label{The proof of the A's form }
&\mathbf{A}={\mathbf{U}_1}^\ast\mathbf{B}{\mathbf{U}_1}^\dag+{\mathbf{U}_1}^\ast\mathbf{C}{\mathbf{U}_2}^\dag.
\end{align}
where, $\mathbf{B}\in\mathbb{C}^{r\times r}$, $\mathbf{C}\in\mathbb{C}^{r\times (K-r)}$.

\begin{proof}
Since $\mathbf{U}=[{\mathbf{U}_1},{\mathbf{U}_2}]$ is a unitary matrix, the following conclusions can be drawn by the definition of the unitary matrix: ${\mathbf{U}_1}\bot{\mathbf{U}_2}$, ${\mathbf{U}_1}^\dag{\mathbf{U}_2}=0 $, ${\mathbf{U}_1}^\dag{\mathbf{U}_1}={\mathbf{I}_r}$, and ${\mathbf{U}_2}^\dag{\mathbf{U}_2}={\mathbf{I}_{K-r}}$. Let $\mathbf{A}$ denote the following block matrix
\begin{align}
\label{The proof of the A's form 1}
\mathbf{A}&={\mathbf{U}}^\ast \begin{bmatrix}
\mathbf{B}&\mathbf{C}\\
\mathbf{D}&\mathbf{E}
\end{bmatrix}{\mathbf{U}}^\dag\nonumber\\
&=[{\mathbf{U}_1}^\ast,{\mathbf{U}_2}^\ast]\begin{bmatrix}
\mathbf{B}&\mathbf{C}\\
\mathbf{D}&\mathbf{E}
\end{bmatrix}\begin{bmatrix}{\mathbf{U}_1}^\dag\\
{\mathbf{U}_2}^\dag \end{bmatrix}\\
&={\mathbf{U}_1}^\ast\mathbf{B}{\mathbf{U}_1}^\dag+{\mathbf{U}_1}^\ast\mathbf{C}{\mathbf{U}_2}^\dag+{\mathbf{U}_2}^\ast\mathbf{D}{\mathbf{U}_1}^\dag+{\mathbf{U}_2}^\ast\mathbf{E}{\mathbf{U}_2}^\dag.\nonumber
\end{align}
where, $\mathbf{B}\in\mathbb{C}^{r\times r}$, $\mathbf{C}\in\mathbb{C}^{r\times (K-r)}$, $\mathbf{D}\in\mathbb{C}^{(K-r)\times r}$, $\mathbf{E}\in\mathbb{C}^{(K-r)\times (K-r)}$. Moreover, there is the following formula
\begin{align}
\label{The proof of the A's form 2}
{\mathbf{U}}^\dag[{{\bf{h}}_S},{{\bf{h}}_J},{{\bf{h}}_D},{{\bf{h}}_E}]&={\mathbf{U}}^\dag{\bf{U}}{\bf{\Omega}}{\bf{V}}^\dag={\bf{\Omega}}{\bf{V}}^\dag
=\begin{bmatrix}{\hat{\mathbf{\Omega}}}\\
\mathbf{0}\end{bmatrix}{\bf{V}}^\dag
=\begin{bmatrix}{\hat{\mathbf{\Omega}}}{\bf{V}}^\dag\\
\mathbf{0}\end{bmatrix}.
\end{align}
where ${\hat{\mathbf{\Omega}}}$ is a matrix consisting of the first $r$ rows of ${\bf{\Omega}}$. At the same time, there is another way to write the above formula.
\begin{align}
\label{The proof of the A's form 3}
{\mathbf{U}}^\dag[{{\bf{h}}_S},{{\bf{h}}_J},{{\bf{h}}_D},{{\bf{h}}_E}]&=\begin{bmatrix}{\mathbf{U}_1}^\dag\\
{\mathbf{U}_2}^\dag \end{bmatrix}[{{\bf{h}}_S},{{\bf{h}}_J},{{\bf{h}}_D},{{\bf{h}}_E}]\\
&=\begin{bmatrix}{\mathbf{U}_1}^\dag{{\mathbf{h}}_S},{\mathbf{U}_1}^\dag{{\mathbf{h}}_J},{\mathbf{U}_1}^\dag{{\mathbf{h}}_D},{\mathbf{U}_1}^\dag{{\mathbf{h}}_E}\\
{\mathbf{U}_2}^\dag{{\mathbf{h}}_S},{\mathbf{U}_2}^\dag{{\mathbf{h}}_J},{\mathbf{U}_2}^\dag{{\mathbf{h}}_D},{\mathbf{U}_2}^\dag{{\mathbf{h}}_E} \end{bmatrix}.\nonumber
\end{align}

By comparing formulas \eqref{The proof of the A's form 2} and \eqref{The proof of the A's form 3}, there are several conclusions, which are ${\mathbf{U}_2}^\dag{\mathbf{h}_S}=0$, ${\mathbf{U}_2}^\dag{\mathbf{h}_J}=0$, ${\mathbf{U}_2}^\dag{\mathbf{h}_D}=0$, and ${\mathbf{U}_2}^\dag{\mathbf{h}_E}=0$. Considering these conclusions, by reasoning, a formula can be obtained as follows
\begin{align}
\label{The proof of the A's form 4}
|{{\mathbf{h}}_D}^T{\mathbf{A}}{{\mathbf{h}}_S}{|^2}&=|{{\mathbf{h}}_D}^T({\mathbf{U}_1}^\ast\mathbf{B}{\mathbf{U}_1}^\dag+{\mathbf{U}_1}^\ast\mathbf{C}{\mathbf{U}_2}^\dag\nonumber\\
&+{\mathbf{U}_2}^\ast\mathbf{D}{\mathbf{U}_1}^\dag+{\mathbf{U}_2}^\ast\mathbf{E}{\mathbf{U}_2}^\dag){{\mathbf{h}}_S}{|^2}\nonumber\\
&=|{{\mathbf{h}}_D}^T{\mathbf{U}_1}^\ast\mathbf{B}{\mathbf{U}_1}^\dag{{\mathbf{h}}_S}{|^2}.
\end{align}

The same as the above analysis, there are
\begin{align}
\label{The proof of the A's form 5}
&|{{\mathbf{h}}_D}^T{\mathbf{A}}{{\mathbf{h}}_J}{|^2}=|{{\mathbf{h}}_D}^T{\mathbf{U}_1}^\ast\mathbf{B}{\mathbf{U}_1}^\dag{{\mathbf{h}}_J}{|^2}\\
&|{{\mathbf{h}}_E}^T{\mathbf{A}}{{\mathbf{h}}_S}{|^2}=|{{\mathbf{h}}_E}^T{\mathbf{U}_1}^\ast\mathbf{B}{\mathbf{U}_1}^\dag{{\mathbf{h}}_S}{|^2}\\
&|{{\mathbf{h}}_E}^T{\mathbf{A}}{{\mathbf{h}}_J}{|^2}=|{{\mathbf{h}}_E}^T{\mathbf{U}_1}^\ast\mathbf{B}{\mathbf{U}_1}^\dag{{\mathbf{h}}_J}{|^2}.
\end{align}

Observing the above formulas, they are independent of $\mathbf{C}$, $\mathbf{D}$ and $\mathbf{E}$. Similarly, based on previous equations, there is the following formula
\begin{align}
\label{The proof of the A's form 6}
\lVert{\mathbf{A}}{{\mathbf{h}}_S}{\rVert}^2&=\lVert{\mathbf{U}_1}^\ast\mathbf{B}{\mathbf{U}_1}^\dag{{\mathbf{h}}_S}+{\mathbf{U}_2}^\ast\mathbf{B}{\mathbf{U}_1}^\dag{{\mathbf{h}}_S}{\rVert}^2\nonumber\\
&=\lVert\mathbf{B}{\mathbf{U}_1}^\dag{{\mathbf{h}}_S}{\rVert}^2+\lVert\mathbf{D}{\mathbf{U}_1}^\dag{{\mathbf{h}}_S}{\rVert}^2.
\end{align}

Considering that ${\mathbf{U}_1}^\dag{\mathbf{U}_2}=0$, ${\mathbf{U}_1}^\dag{\mathbf{U}_1}={\mathbf{I}_r}$ and ${\mathbf{U}_2}^\dag{\mathbf{U}_2}={\mathbf{I}_{K-r}}$ are established, for the same reason, there are the following equations
\begin{align}
\label{The proof of the A's form 7}
&\lVert{\mathbf{A}}{{\mathbf{h}}_J}{\rVert}^2=\lVert\mathbf{B}{\mathbf{U}_1}^\dag{{\mathbf{h}}_J}{\rVert}^2+\lVert\mathbf{D}{\mathbf{U}_1}^\dag{{\mathbf{h}}_J}{\rVert}^2\\
&\lVert{\mathbf{A}}{{\mathbf{h}}_D}{\rVert}^2=\lVert\mathbf{B}{\mathbf{U}_1}^\dag{{\mathbf{h}}_D}{\rVert}^2+\lVert\mathbf{D}{\mathbf{U}_1}^\dag{{\mathbf{h}}_D}{\rVert}^2\\
&\lVert{{\mathbf{h}}_D}^T{\mathbf{A}}{\rVert}^2=\lVert{{\mathbf{h}}_D}^T{\mathbf{U}_1}^\ast\mathbf{B}{\rVert}^2+\lVert{{\mathbf{h}}_D}^T{\mathbf{U}_1}^\ast\mathbf{C}{\rVert}^2\\
&\lVert{{\mathbf{h}}_E}^T{\mathbf{A}}{\rVert}^2=\lVert{{\mathbf{h}}_E}^T{\mathbf{U}_1}^\ast\mathbf{B}{\rVert}^2+\lVert{{\mathbf{h}}_E}^T{\mathbf{U}_1}^\ast\mathbf{C}{\rVert}^2\\
&tr(\mathbf{A}{\mathbf{A}}^\dag)=tr(\mathbf{B}{\mathbf{B}}^\dag)+tr(\mathbf{C}{\mathbf{C}}^\dag)+tr(\mathbf{D}{\mathbf{D}}^\dag)+tr(\mathbf{E}{\mathbf{E}}^\dag).
\end{align}

Based on the above analysis, we analyze optimization problem \eqref{case2:11}. Since the limited length of this paper, let $\delta=|{{\mathbf{h}}_D}^T{\mathbf{U}_1}^\ast\mathbf{B}{\mathbf{U}_1}^\dag{{\mathbf{h}}_J}{|^2}+\lVert{{\mathbf{h}}_D}^T{\mathbf{U}_1}^\ast\mathbf{B}{\rVert}^2+\lVert{{\mathbf{h}}_D}^T{\mathbf{U}_1}^\ast\mathbf{C}{\rVert}^2+1$,  then it can be written as
\begin{align}
\label{The proof of the A's form 8}
\underset{\bf{A}}{\max}\ \
&\frac{|{{\mathbf{h}}_D}^T{\mathbf{U}_1}^\ast\mathbf{B}{\mathbf{U}_1}^\dag{{\mathbf{h}}_S}{|^2}}{\delta}\nonumber \\
\text{s.t.} \ \ &{\mathbf{P}_S}\lVert\mathbf{B}{\mathbf{U}_1}^\dag{{\mathbf{h}}_S}{\rVert}^2+{\mathbf{P}_S}\lVert\mathbf{D}{\mathbf{U}_1}^\dag{{\mathbf{h}}_S}{\rVert}^2+{\mathbf{P}_J}\lVert\mathbf{B}{\mathbf{U}_1}^\dag{{\mathbf{h}}_J}{\rVert}^2+{\mathbf{P}_J}\lVert\mathbf{D}{\mathbf{U}_1}^\dag{{\mathbf{h}}_J}{\rVert}^2+tr(\mathbf{B}{\mathbf{B}}^\dag)\nonumber\\
&+tr(\mathbf{C}{\mathbf{C}}^\dag)+tr(\mathbf{D}{\mathbf{D}}^\dag)+tr(\mathbf{E}{\mathbf{E}}^\dag)\le P_{R,\max}\nonumber\\
&|{{\mathbf{h}}_E}^T{\mathbf{U}_1}^\ast\mathbf{B}{\mathbf{U}_1}^\dag{{\mathbf{h}}_S}{|^2}{P_S} + |{{\mathbf{h}}_E}^T{\mathbf{U}_1}^\ast\mathbf{B}{\mathbf{U}_1}^\dag{{\mathbf{h}}_J}{|^2}{P_J}+\lVert{{\mathbf{h}}_E}^T{\mathbf{U}_1}^\ast\mathbf{B}{\rVert}^2+\lVert{{\mathbf{h}}_E}^T{\mathbf{U}_1}^\ast\mathbf{C}{\rVert}^2\ge Q \nonumber \\
&\frac{|{{\mathbf{h}}_D}^T{\mathbf{U}_1}^\ast\mathbf{B}{\mathbf{U}_1}^\dag{{\mathbf{h}}_J}{|^2}{P_J}}{|{{\mathbf{h}}_D}^T{\mathbf{U}_1}^\ast\mathbf{B}{\mathbf{U}_1}^\dag{{\mathbf{h}}_S}{|^2}{P_S}}<1.
\end{align}

Obviously, the second and third constraints and the objective function are independent of the matrix $\mathbf{D}$, $\mathbf{E}$, and only related to the matrix $\mathbf{B}$, $\mathbf{C}$. Given matrices $\mathbf{B}$, $\mathbf{C}$, when $\mathbf{D}=0$, $\mathbf{E}=0$, the second and third constraints remain unchanged, and the feasible domain determined by the first constraint will expand. Therefore, the optimal solution for optimization problem \eqref{case2:11} is only related to the matrix $\mathbf{B}$, $\mathbf{C}$. For matrix $\mathbf{A}$, the following equation exists, which is $\mathbf{A}={\mathbf{U}_1}^\ast\mathbf{B}{\mathbf{U}_1}^\dag+{\mathbf{U}_1}^\ast\mathbf{C}{\mathbf{U}_2}^\dag$.

\end{proof}

\section{The proof of the theorem 2}
Assuming that $(\mathbf{X},b,c)$ is one of the optimal solutions for optimization problem \eqref{case2:19}, then there is always a $(\mathbf{X}^\star,b^\star,c^\star)$ satisfying $rank(\mathbf{X}^\star)=1$ and we can find it in polynomial time.
\begin{proof}
The proof of the proposition will start from two aspects, ${{\bf{h}}_S} \parallel {{\bf{h}}_J}$ and ${{\bf{h}}_S} \nparallel {{\bf{h}}_J}$.
\subsubsection{${{\bf{h}}_S} \parallel {{\bf{h}}_J}$}
Let ${{\bf{h}}_1}=t_1{{\bf{h}}}$, ${{\bf{h}}_2}=t_2{{\bf{h}}}$, ${{\bf{h}}_3}=m_1\hat{{{\bf{h}}}}$, ${{\bf{h}}_4}=m_2\hat{{{\bf{h}}}}$, ${{\bf{H}}_2}=r_1{{\bf{H}}}$, ${{\bf{H}}_3}=r_2{{\bf{H}}}$, where $t_1$, $t_2$, $m_1$, $m_2$, $r_1$, $r_2$ are constants, then \eqref{case2:18} can be rewritten as
\begin{align}
\label{case2:20}
&\underset{\mathbf{X},b^2,c^2}{\max}\ \
tr({|{t_1}{|^2}{\mathbf{h}}{\mathbf{h}}^\dag \mathbf{X}}) \nonumber \\
&\text{s.t.} \ \ {\rho_0}{c^2}+tr({\rho_1}{{\bf{H}}_1}{{\bf{H}}_1}^\dag \mathbf{X})+{\rho_2}{b^2}\le1,tr(|{t_2}{|^2}{\mathbf{h}}{\mathbf{h}}^\dag \mathbf{X})\le{c^2}\nonumber\\
&tr(|{t_1}{|^2}{\mathbf{h}}{\mathbf{h}}^\dag \mathbf{X})\ge \frac{{\rho_0} c^2}{\varepsilon},\mathbf{X}\succeq 0,rank(\mathbf{X})\le1\\
&tr[(|{m_1}{|^2}{\hat{{\mathbf{h}}}}{\hat{{\mathbf{h}}}}^\dag{P_S}+|{m_2}{|^2}{\hat{{\mathbf{h}}}}{\hat{{\mathbf{h}}}}^\dag{P_J}+{\mathbf{H}_4}{\mathbf{H}_4}^\dag{\sigma_R}^2)\mathbf{X}]\ge {Q}{b^2} \nonumber \\
&tr[(|{r_1}{|^2}{{\mathbf{H}}}{{\mathbf{H}}}^\dag{P_S}+|{r_2}{|^2}{{\mathbf{H}}}{{\mathbf{H}}}^\dag{P_J}+{\sigma_R}^2{{\mathbf{I}}})\mathbf{X}]\le {P_{R,\max}}{b^2}\nonumber
\end{align}

Observing \eqref{case2:20}, it is easy to know that the second and third constraints can be classified as one constraint. Therefore, the problem reduces one constraint and becomes four constraints. It can be seen that the SDP problem with four constraints satisfies the rank-1 constraint \cite{t13}.

\subsubsection{${{\bf{h}}_S} \nparallel {{\bf{h}}_J}$}
The K-K-T equation is used to study the properties of the rank of the optimal solution obtained in \eqref{case2:18}. The Lagrangian function of \eqref{case2:18} is as follows
\begin{align}
\label{case2:21}
&\mathcal{L}(\mu,\nu,m,n,\omega,\mathbf{Y})=tr[((1+m){{\bf{h}}_1} {{\bf{h}}_1}^\dag-\mu {\rho_1}\mathbf{H}_1\mathbf{H}_1^\dag+\nu {{\bf{h}}_2} {{\bf{h}}_2}^\dag+n(P_S{{\bf{h}}_3}{{\bf{h}}_3}^\dag+P_J{{\bf{h}}_4} {{\bf{h}}_4}^\dag\nonumber \\
&+\mathbf{H}_4\mathbf{H}_4^\dag{\sigma_R}^2)-\omega(P_S\mathbf{H}_2\mathbf{H}_2^\dag+P_J\mathbf{H}_3\mathbf{H}_3^\dag+{\sigma_R}^2\mathbf{I})+\mathbf{Y})\mathbf{X}]-\mu {\rho_0}c^2-\mu {\rho_2}b^2+\mu-\nu c^2\nonumber \\
&-\frac{{\rho_0}mc^2}{\varepsilon}-nQb^2+\omega P_{R, \max }b^2
\end{align}
where, $\mu,\nu,m,n,\omega$ and $\mathbf{Y}$ are Lagrangian multipliers for the constraint of problem \eqref{case2:18}, respectively. Since the constraint corresponding to $\nu$ is an equation, $\nu=0$ or $\nu>0$ is not set for subsequent analysis. To obtain a finite value for $\mathcal{L}$ under the conditions of arbitrary $\mu \ge0,\nu,m\ge0,n\ge0,\omega\ge0$ and $\bf{Y}\succeq 0$ , the following condition must be satisfied.
\begin{align}
\label{case2:22}
&((1+m){{\bf{h}}_1} {{\bf{h}}_1}^\dag-\mu {\rho_1}\mathbf{H}_1\mathbf{H}_1^\dag+\nu {{\bf{h}}_2} {{\bf{h}}_2}^\dag+n(P_S{{\bf{h}}_3} {{\bf{h}}_3}^\dag+P_J{{\bf{h}}_4} {{\bf{h}}_4}^\dag+\mathbf{H}_4\mathbf{H}_4^\dag{\sigma_R}^2)\nonumber \\
&-\omega(P_S\mathbf{H}_2\mathbf{H}_2^\dag+P_J\mathbf{H}_3\mathbf{H}_3^\dag+{\sigma_R}^2\mathbf{I})+\bf{Y}=0
\end{align}

Through the strong duality theorem, $\bf{X}$ and $\bf{Y}$ can be obtained to satisfy the complementary relaxation relationship
\begin{equation}
\label{case2:24}
\bf{X}^\star \bf{Y}^\star=0
\end{equation}

Considering \eqref{case2:22} and \eqref{case2:24}, there is a formula as
\begin{align}
\label{case2:25}
&rank(\mathbf{X}^\star)=rank[(\mu^\star {\rho_1}\mathbf{H}_1\mathbf{H}_1^\dag+\omega^\star(P_S\mathbf{H}_2\mathbf{H}_2^\dag+P_J\mathbf{H}_3\mathbf{H}_3^\dag+{\sigma_R}^2\mathbf{I}))\bf{X}^\star]\nonumber\\
&=rank[((1+m^\star){{\bf{h}}_1} {{\bf{h}}_1}^\dag+n^\star(P_S{{\bf{h}}_3} {{\bf{h}}_3}^\dag+P_J{{\bf{h}}_4} {{\bf{h}}_4}^\dag+\mathbf{H}_4\mathbf{H}_4^\dag{\sigma_R}^2)+\nu^\star {{\bf{h}}_2} {{\bf{h}}_2}^\dag)\bf{X}^\star]\nonumber\\
&\le rank[((1+m^\star){{\bf{h}}_1} {{\bf{h}}_1}^\dag+n^\star(P_S{{\bf{h}}_3} {{\bf{h}}_3}^\dag+P_J{{\bf{h}}_4} {{\bf{h}}_4}^\dag+\mathbf{H}_4\mathbf{H}_4^\dag{\sigma_R}^2)+\nu^\star {{\bf{h}}_2} {{\bf{h}}_2}^\dag)]
\end{align}

If the fourth constraint inequality is strictly established, the complementary relaxation relationship can be obtained
\begin{align}
\label{case2:27}
&n^\star[tr[({{\mathbf{h}}_3} {{\mathbf{h}}_3}^\dag{P_S}+{{\mathbf{h}}_4} {{\mathbf{h}}_4}^\dag{P_J}+\mathbf{H}_4\mathbf{H}_4^\dag{\sigma_R}^2)\mathbf{X}]-{Q}{b^2}]=0\\
&\nu^\star[tr({{\mathbf{h}}_2}{{\mathbf{h}}_2}^\dag \mathbf{X})-{c^2}]=0
\end{align}

Therefore, the optimal value of $n^\star$, $\nu^\star$ are $n^\star=0$ and $\nu^\star=0$. When the first and fifth constraints are strictly equal, the objective function has a better value. Thus, through the complementary relaxation theorem, $\mu^\star$ and $\omega^\star$ are generally not zero. In this case, \eqref{case2:25} can be converted to
\begin{align}
\label{case2:26}
&rank(\mathbf{X}^\star)\le rank[(1+m^\star){{\mathbf{h}}_1} {{\mathbf{h}}_1}^\dag]\le 1
\end{align}

Hence, the optimal solution $\bf{X}^\star$ satisfies the rank-1 constraint.
\end{proof}

\ifCLASSOPTIONcaptionsoff
  \newpage
\fi



%\bibliographystyle{IEEEtranTCOM}
\bibliography{IEEEabrv,references}

%\begin{biography}[{\includegraphics[width=1in,height=1.25in,clip,keepaspectratio]{mshell}}]{Michael Shell}


%\vfill


\end{document}